\documentclass[conference]{IEEEtran}
\IEEEoverridecommandlockouts
\usepackage{cite}
\usepackage{amsmath,amssymb,amsfonts}
\usepackage{algorithmic}
\usepackage{graphicx}
\usepackage{textcomp}

\usepackage{bm}
\usepackage{mathtools}
\usepackage[table]{xcolor}
\usepackage{amsthm,hyperref}
\usepackage{caption}
\usepackage{subcaption}

\usepackage{xcolor,pifont}

\usepackage{accents}
\newcommand{\ubar}[1]{\underaccent{\bar}{#1}}

\newtheorem{theorem}{Theorem}

\definecolor{darkgreen}{rgb}{0.0, 0.4, 0.0}
\definecolor{mygreen}{rgb}{0.0, 0.75, 0.0}
\definecolor{hunterblazeorange}{rgb}{1.0, 0.5098, 0.0}
\definecolor{lightgray}{gray}{0.9}

\DeclareMathAlphabet\mathbfcal{OMS}{cmsy}{b}{n}

\DeclareMathOperator*{\st}{subject~to}

\DeclareMathOperator*{\argmin}{arg\,min}

\newcommand{\norm}[1]{\left\lVert#1\right\rVert}
\newcommand{\normtwo}[1]{\left\lVert#1\right\rVert_{2}}

\newtheorem{lemma}{Lemma}
\newtheorem{proposition}{Proposition}
\newtheorem{corollary}{Corollary}

\newtheorem{assumption}{Assumption}
\begin{document}

\title{\LARGE On the Guarantees of Minimizing Regret in Receding Horizon}

\author{Andrea Martin, Luca Furieri, Florian D\"orfler, John Lygeros, and Giancarlo Ferrari-Trecate
\thanks{Research supported by the Swiss National Science Foundation under the NCCR Automation (grant agreement 51NF40\textunderscore 80545) and the Ambizione grant PZ00P2
\textunderscore 208951.}
\thanks{A. Martin, L. Furieri, and G. Ferrari-Trecate are with the Institute of Mechanical Engineering, EPFL, Switzerland. E-mail addresses: \{andrea.martin, luca.furieri, giancarlo.ferraritrecate\}@epfl.ch.}
\thanks{F. D\"orfler and J. Lygeros are with the Department of Information Technology and Electrical Engineering, ETH Z\"urich, Switzerland. E-mail addresses: \{dorfler, jlygeros\}@ethz.ch.}}

\maketitle

\begin{abstract}
    Towards bridging classical optimal control and online learning, regret minimization has recently been proposed as a control design criterion. This competitive paradigm penalizes the loss relative to the optimal control actions chosen by a clairvoyant policy, and allows tracking the optimal performance in hindsight no matter how disturbances are generated. In this paper, we propose the first receding horizon scheme based on the repeated computation of finite horizon regret-optimal policies, and we establish stability and safety guarantees for the resulting closed-loop system. Our derivations combine novel monotonicity properties of clairvoyant policies with suitable terminal ingredients. We prove that our scheme is recursively feasible, stabilizing, and that it achieves bounded regret relative to the infinite horizon clairvoyant policy. Last, we show that the policy optimization problem can be solved efficiently through convex-concave programming. Our numerical experiments show that minimizing regret can outperform standard receding horizon approaches when the disturbances poorly fit classical design assumptions – even when the finite horizon planning is recomputed less frequently.
\end{abstract}

\section{Introduction}
Many feedback control methods aim to optimize a performance measure with respect to a specific class of exogenous perturbations. For instance, classical $\mathcal{H}_2$ and $\mathcal{H}_\infty$ control paradigms assume that a stochastic or adversarial disturbance process drives the system dynamics, respectively, and minimize the expected or worst-case control cost accordingly \cite{hassibi1999indefinite}. As these control policies are tailored to the assumed nature of the exogenous perturbations, they may incur significant loss or be conservative if the true disturbances do not fulfill the modeling assumptions \cite{goel2022power, karapetyan2022regret}.

For striking a balance between nominal performance and robustness, several approaches have been proposed, including risk-sensitive control \cite{jacobson1973optimal}, mixed $\mathcal{H}_2 {/} \mathcal{H}_\infty$ \cite{bernstein1988lqg, doyle1989optimal, rotea1991h2, scherer1995mixed}, and adversarially robust control \cite{lee2022performance}. However, these control strategies often settle upon a robustness level selected a priori, with no regard for the disturbance sequence observed during online operation. The resulting policies therefore lack the adaptivity that is required to take full advantage of dynamic environments.

In the computer science literature, the design of sequential decision-making algorithms that learn from experience has traditionally been approached from the perspective of regret minimization \cite{shalev2011online, hazan2016introduction}. This framework encourages agents to dynamically adjust their strategy – based on information deduced from previous rounds – by penalizing the loss relative to the optimal policy in hindsight. Regret-based methods offer attractive guarantees in terms of the performance of the best policy an agent could have counterfactually played, which hold independently of the stochastic or adversarial nature of the disturbances.

While classical online learning theory has mostly considered memoryless environments (see, for instance, the rich literature on bandit problems \cite{lattimore2020bandit}), recent years have witnessed increasing interest in applying modern statistical and algorithmic techniques to settings with dynamics. Initiated by \cite{abbasi2011regret}, several learning algorithms have been proposed for adaptively controlling a linear dynamical system driven by a stochastic \cite{dean2018regret, cohen2019learning, lale2020logarithmic} or adversarial \cite{agarwal2019online, foster2020logarithmic, hazan2020nonstochastic, simchowitz2020improper} disturbance process. The methods above attain sublinear regret against the classes of state-feedback and disturbance-action policies, that is, the average difference between the cost they incur and that of the best fixed strategy converges to zero over time.
On the other hand, policy regret methods offer no guarantee that this static benchmark incurs a low control cost, leading to possibly loose performance certificates \cite{hazan2009efficient}. For instance, the authors of \cite{goel2023regret} have shown that no single static feedback controller can perform well in a scenario where the disturbances alternate between being drawn according to a well-behaved stochastic process and being chosen adversarially. Hence, any online algorithm that tries to learn the best fixed state-feedback law will also incur a high control cost.

The drawbacks of policy regret methods have motivated recent literature to design algorithms that compete against a series of time-varying benchmarks \cite{zinkevich2003online}. Gradient-based methods that achieve low dynamic regret against the class of disturbance-action policies have been proposed, e.g., in \cite{zhao2020dynamic} and \cite{zhou2023efficient}. A control law that exactly minimizes the loss relative to the globally optimal sequence of control actions in hindsight has been computed in \cite{goel2023regret, sabag2021regret} via reductions to classical $\mathcal{H}_\infty$ synthesis and Nehari noncausal approximation problems \cite{nehari1957bounded}. Algorithms that compete against time-varying benchmarks have been shown to effectively interpolate between the performance of classical $\mathcal{H}_2$ and $\mathcal{H}_\infty$ controllers in a variety of applications of interest, ranging from longitudinal motion control of a helicopter to control of a wind energy conversion system \cite{sabag2022optimal}. However, the resulting control policies cannot comply with the safety requirements of many practical applications due to a lack of provable robustness guarantees.

With the aim of allowing a reliable deployment of regret minimization methods, in earlier work \cite{martin2022safe} we leveraged the system level synthesis (SLS) framework \cite{wang2019system} to formulate safe regret-optimal control problems as semidefinite optimization programs; a similar approach was adopted in \cite{didier2022system} to tackle the case where the exogenous perturbations satisfy instantaneous ellipsoidal bounds. The poor scalability of semidefinite optimization, however, hinders the widespread application of the results in \cite{martin2022safe, didier2022system, didier2022generalised, martin2022follow, martin2023regret}; efficiently minimizing dynamic regret over an infinite
horizon – while guaranteeing stability and safety of the closed-loop system – remains an open challenge.

\emph{Contributions:}
Motivated by the ability of robust model predictive control (MPC) to handle constraints on the physical variables of a system \cite{rawlings2017model}, we propose a stabilizing receding horizon scheme based on the repeated computation of finite horizon safe regret-optimal feedback policies. Unlike classical architectures with $\mathcal{H}_2$ and $\mathcal{H}_\infty$ objectives, static policies do not attain minimal dynamic regret \cite{sabag2021regret} in this setting. Further, the cost-to-go function – commonly employed as a control Lyapunov function in proving asymptotic stability of the closed-loop system \cite{mayne2000constrained} – cannot be computed as a quadratic function of the state only.
To address the issues above, we shape the terminal ingredients based on an auxiliary $\mathcal{H}_\infty$ policy that induces a quadratic upper bound on the regret to go. Moreover, to compete against the optimal benchmark along each planning horizon, we repeatedly compare with finite horizon clairvoyant policies. Yet, we establish regret guarantees of our receding horizon control law relative to the infinite horizon clairvoyant policy. %
As a result of our analysis, we show that the proposed strategy enjoys recursive feasibility and guarantees $\ell_2$-stability of the closed-loop system. %
A main challenge has been that the same guarantees cannot be established by reducing the regret minimization to an $\ell_2$-gain attenuation problem through the construction of an auxiliary system \cite{goel2023regret} combined with receding horizon $\mathcal{H}_\infty$ control \cite{kim2004disturbance, goulart2009control}. Indeed, the dynamics of the augmented system resulting from a series of Riccati recursions are time-varying in general, even if the dynamics of the original system are not \cite{goel2023regret}. Last, we present numerical experiments to assess the efficacy of dynamic regret as a receding horizon design criterion.

\emph{Organization:} Section \ref{sec:problem_statement_preliminaries} formalizes the competitive problem of interest, reviews the system level approach to controller synthesis, and recalls useful properties of the clairvoyant optimal policy. Section \ref{sec:finite_horizon} studies regret minimization over a finite horizon, summarizing and extending the results of our previous work \cite{martin2022safe}. Section \ref{sec:infinite_horizon} establishes novel monotonicity properties of the clairvoyant optimal policy, and presents conditions that guarantee recursive feasibility and bounded regret of the proposed receding horizon control scheme. Section \ref{sec:numerical} discusses efficient computational methods and collects our numerical results. Finally, Section \ref{sec:conclusion} summarizes the contributions and outlines directions for future research.

\emph{Notation:}
We denote the sets of natural, integer and real numbers by $\mathbb{N}$, $\mathbb{Z}$ and $\mathbb{R}$, respectively. We use $\mathbb{R}_{\geq 0}$ to denote the non-negative reals and $\mathbb{I}_{[a,b]}$ to denote the set of integers in the interval $[a, b] \subseteq \mathbb{R}$. We use lower and upper case letters such as $x$ and $A$ to denote vectors and matrices, respectively, and lower and upper case boldface letters such as $\mathbf{x}$ and $\mathbf{A}$ to denote finite horizon signals and operators, respectively. We use calligraphic letters such as $\mathcal{X}$ to denote sets. We use $I$ to denote the identity matrix when dimensions are clear from the context. We denote by $\lambda_{\operatorname{max}}(A)$ and $\lambda_{\operatorname{min}}(A)$ the maximum and minimum eigenvalue of the matrix $A$, respectively. We write $A \succ 0$ or $A \succeq 0$ if the symmetric matrix $A$ is positive definite or positive semidefinite, respectively. Inequalities involving vectors are applied element-wise. We use $\otimes$ to denote the Kronecker product. We use $A^{[i,j]}$ to denote the $n \times n$ submatrix obtained extracting the $(i,j)$ block entry of a block matrix $A$. Similarly, we write $A^{[i_1:i_2,j]}$ to denote the $(i_2 - i_1 + 1)n \times n$ submatrix obtained stacking horizontally the $(i,j)$ block entries of a block matrix $A$, where the index $i$ ranges from $i_1$ to $i_2$, inclusive.

\section{Problem Statement and Preliminaries}
\label{sec:problem_statement_preliminaries}
We consider discrete-time linear time-invariant dynamical systems described by the state-space equations
\begin{equation}
    \label{eq:lti_system_dynamics}
    x_{t+1} = A x_t + B u_t + w_t\,,
\end{equation}
where $x_t \in \mathbb{R}^n$, $u_t \in \mathbb{R}^m$, and $w_t \in \mathbb{R}^n$ are the system state, the control input, and an exogenous signal, respectively. We do not make any assumptions about the statistical distribution of the disturbance process, and we let the realizations $w_t$ be drawn arbitrarily, potentially by an adversary, from a set $\mathcal{W}$. As standard in the robust MPC literature \cite{rawlings2017model}, we assume that $\mathcal{W}$ is convex, compact and contains the origin in its interior.

Our objective is to construct a stabilizing control policy with infinite horizon safety and performance guarantees. To do so, we repeatedly optimize the system behavior over a finite planning horizon of length $T \in \mathbb{N}$, and only apply the first $s \in \mathbb{I}_{[1, T]}$ control moves in a receding horizon fashion. For compactness, we write the evolution of the system state and input trajectories over $T$ as
\begin{equation*}
    \mathbf{x} = \mathbf{Z}\mathbf{A}\mathbf{x} + \mathbf{Z}\mathbf{B}\mathbf{u} + \bm{\delta}\,,
\end{equation*}
where $\mathbf{Z}$ denotes the block-downshift operator, that is, %
\begin{equation*}
    \mathbf{Z} = \begin{bmatrix}
        0_{n \times n} & \cdots & \cdots & \cdots & 0_{n \times n}\\
        I_n & \ddots & & & \vdots\\
        0_{n \times n} & \ddots & \ddots & & \vdots\\
        \vdots & \ddots & \ddots & \ddots & \vdots\\
        0_{n \times n} & \cdots & 0_{n \times n} & I_n & 0_{n \times n}
    \end{bmatrix}\,,
\end{equation*}
$\mathbf{A} = I_{T+1} \otimes A$, $\mathbf{B} = \operatorname{col}(I_{T} \otimes B, 0)$, $\bm{\delta} = \operatorname{col}(x_0, \mathbf{w})$ and
\begin{alignat*}{3}
    \mathbf{x} &= \operatorname{col}(x_0, x_1, \dots, x_T) =
    \begin{bmatrix}
        x_0^\top & x_1^\top & \cdots & x_{T}^\top \end{bmatrix}^\top\,,\\
    \mathbf{u} &= \operatorname{col}(u_0, u_1, \dots, u_{T-1})\,, ~ \mathbf{w} = \operatorname{col}(w_0, w_1, \dots, w_{T-1})\,.
\end{alignat*}

To regulate the behavior of system \eqref{eq:lti_system_dynamics}, we consider time-varying dynamic linear feedback control policies of the form
\begin{equation}
    \label{eq:control_policy}
    u_t = \pi_t(x_{0:t}) = K_{t,t} x_t + \cdots + K_{t,0} x_0\,, ~ \forall t \in \mathbb{I}_{[0, T-1]}\,;
\end{equation}
equivalently, $\mathbf{u} = \mathbf{K} \mathbf{x}$, where $\mathbf{K}$ exhibits the following lower block-triangular structure due to causality:
\begin{equation*}
    \mathbf{K} =
    \begin{bmatrix}
    	K_{0,0} & 0_{m \times n} & \dots & 0_{m \times n}\\
    	K_{1,0} & K_{1,1} & \ddots & \vdots\\
    	\vdots & \vdots & \ddots & 0_{m \times n}\\
    	K_{T-1, 0} & K_{T-1, 1}& \cdots & K_{T-1, T-1}
    \end{bmatrix}\,.
\end{equation*}
We note that this choice ensures that minimum dynamic regret over $T$ is attained in unconstrained scenarios \cite{goel2023regret, sabag2021regret}, and, as we will discuss in Section~\ref{sec:numerical}, it also allows for policy optimization via convex programming techniques through a suitable re-parametrization \cite{goulart2006optimization, wang2019system}. We then restrict our attention to control laws that are admissible, namely that guarantee compliance with the following safety constraints:
\begin{align}
    \label{eq:safety_constraints}
    (x_t, u_t) \in \mathcal{Z}\,, ~
    x_T \in \mathcal{Z}_f\,, ~ \forall w_t \in \mathcal{W}\,, ~ \forall t \in \mathbb{I}_{[0, T-1]}\,,
\end{align}
where $\mathcal{Z} \subseteq \mathbb{R}^n \times \mathbb{R}^m$ and $\mathcal{Z}_f \subseteq \mathbb{R}^n$ are closed convex sets that contain the origin in their interior. We denote the set of all admissible policies $\bm{\pi}_T = (\pi_0, \dots, \pi_{T-1})$ of the form \eqref{eq:control_policy} as $\bm{\Pi}_T(x_0)$. Similarly, we denote the set of initial conditions for which an admissible linear control policy exists as $\mathcal{X}_T = \{x \in \mathbb{R}^n : \bm{\Pi}_T(x) \neq \emptyset\}$.

For any possible realization of the disturbance sequence $\mathbf{w}$, we measure the control cost that an admissible policy $\bm{\pi}_T \in \bm{\Pi}_T(x_0)$ incurs along the planning horizon by
\begin{equation}
    \label{eq:lqr_cost}
    J_T(\bm{\pi}_T, \bm{\delta}) = \sum_{k = 0}^{T-1} x_{k}^\top Q x_{k} + u_k^\top R u_k\,,%
\end{equation}
where $Q \succeq 0$ and $R \succ 0$ are the state and control stage cost matrices, respectively.  Since \eqref{eq:lqr_cost} depends on yet unknown disturbance realizations, computing the cost-minimizing sequence of control actions given causal information becomes an ill-posed problem. To break this deadlock, classical $\mathcal{H}_2$ and $\mathcal{H}_\infty$ paradigms assume that $\mathbf{w}$ follows a known probability distribution and that $\mathbf{w}$ is selected adversarially, respectively \cite{hassibi1999indefinite}. Then, $\mathcal{H}_2$ methods optimize the average cost $\mathbb{E}[J_T(\bm{\pi}_T, \bm{\delta}) + V_f(x_T)]$, where $V_f(x_T) = x_T^\top P x_T$ denotes a terminal cost weighted by $P \succeq 0$, whereas $\mathcal{H}_\infty$ methods minimize the performance level $\gamma \in \mathbb{R}_{\geq 0}$ for which the bound
\begin{equation}
    \label{eq:h_infinity_bound}
    J_T(\bm{\pi}_T, \bm{\delta}) + V_f(x_T) \leq \gamma^2 \norm{\mathbf{w}}_2^2 + \beta(x_0, \gamma)\,,
\end{equation}
with $\beta(x_0, \gamma) \in \mathbb{R}_{\geq 0}$, holds for all $\mathbf{w}$ with bounded energy. Instead, at each planning stage, we consider the objective of safely minimizing the worst-case loss relative to the clairvoyant optimal policy $\bm{\pi}^c_T$. The latter is defined by
\begin{equation}
    \label{eq:clairvoyant_argmin_definition}
    \bm{\pi}^c_T = \operatorname{argmin}_{\bm{\pi} \in \bm{\Pi}_{T}^c} ~ J_T(\bm{\pi}, \bm{\delta}) = \operatorname{argmin}_{\mathbf{u}(\bm{\delta})} ~ J_T(\mathbf{u}, \bm{\delta})\,,
\end{equation}
where $\bm{\Pi}_{T}^c$ denotes the set of all possibly nonlinear and clairvoyant (noncausal) control laws that could be designed if we had complete foreknowledge of the disturbance realizations. As such, policies in $\bm{\Pi}_{T}^c$ yield an ideal, yet unattainable, closed-loop performance that we strive to reproduce by designing a causal policy.

Specifically, given a performance level $\gamma \in \mathbb{R}_{\geq 0}$, let $\mathtt{Reg}_T(\bm{\pi}_T, \gamma, \bm{\delta})$ be defined according to
\begin{equation}
    \label{eq:regret_definition}
    \mathtt{Reg}_T(\cdot) = J_T(\bm{\pi}_T, \bm{\delta}) - J_T(\bm{\pi}^c_T, \bm{\delta}) + V_f(x_T) - \gamma^2 \norm{\mathbf{w}}_2^2\,,
\end{equation}
Then, we aim at synthesizing an admissible control policy $\bm{\pi}_T \in \bm{\Pi}_T(x_0)$ and a non-negative real scalar $\beta(x_0, \gamma)$ such that the dynamic regret bound
\begin{equation}
    \label{eq:dynamic_regret_bound}
    \mathtt{Reg}_T(\bm{\pi}_T, \gamma, \bm{\delta}) \leq \beta(x_0, \gamma)\,,
\end{equation}
holds for all disturbance sequences $\mathbf{w} \in \mathbfcal{W} = \mathcal{W}^{T} = \mathcal{W} \times \cdots \times \mathcal{W}$. Note that, to synthesize a control policy whose dynamic regret has optimal dependence on the energy of the disturbance,\footnote{Dynamic regret bounds are often expressed in terms of some ``regularity'' of the perturbation sequence since dynamic regret scales linearly with time in the worst-case \cite{goel2022power, zhao2020dynamic}.} we aim at fulfilling the regret bound in \eqref{eq:dynamic_regret_bound} with performance level $\gamma \in \mathbb{R}_{\geq 0}$ as small as possible. Finally, we point out that the role of the terminal penalty $V_f(x_T)$ will be to compute a quadratic upper bound to the regret-to-go, and to design, jointly with the terminal constraint set $\mathcal{Z}_f$, a receding horizon control law with guaranteed recursive feasibility and stability properties.

\subsection{System Level Synthesis}
\label{subsec:system_level_synthesis}
We now briefly outline the necessary background on the SLS approach to optimal controller synthesis, and refer to \cite{wang2019system} for a complete discussion. Akin to the Youla parametrization \cite{youla1976modern} and disturbance-feedback controllers \cite{goulart2006optimization}, the SLS approach shifts the synthesis problem from directly designing the controller to shaping the closed-loop maps from the exogenous disturbance to the state and input signals \cite{zheng2020equivalence, furieri2019input, furieri2023near}.

Along the planning horizon $T$, the behavior of the closed-loop system under the feedback interconnection $\mathbf{u} = \mathbf{K} \mathbf{x}$ in \eqref{eq:control_policy} can be described through the relations
    \begin{align}
        \mathbf{x} &= \left(\mathbf{I} - \mathbf{Z}\left(\mathbf{A} +  \mathbf{B}\mathbf{K}\right)\right)^{-1} \bm{\delta} := \bm{\Phi}_x \bm{\delta} = \bm{\Phi}_x^0 x_0 + \bm{\Phi}_x^w \mathbf{w}\,, \label{eq:closed_loop_behavior_x_K}\\
        \mathbf{u} &= \mathbf{K}\bm{\Phi}_x \bm{\delta} := \bm{\Phi}_u \bm{\delta} = \bm{\Phi}_u^0 x_0 + \bm{\Phi}_u^w \mathbf{w}\,, \label{eq:closed_loop_behavior_u_K}
    \end{align}
where $\bm{\Phi}_x = \operatorname{col}(\bm{\Phi}_x^0, \bm{\Phi}_x^w)$ and $\bm{\Phi}_u = \operatorname{col}(\bm{\Phi}_u^0, \bm{\Phi}_u^w)$. By inspection, the expressions above are non-convex in $\mathbf{K}$, yet they are linear in the closed-loop responses $\{\bm{\Phi}_x, \bm{\Phi}_u\}$ that the controller $\mathbf{K}$ induces. Note also that these operators inherit a lower block-triangular causal structure from their definition. The SLS framework then shows that there exists a feedback controller $\mathbf{K}$ such that $\mathbf{x} = \bm{\Phi}_x \bm{\delta}$ and $\mathbf{u} = \bm{\Phi}_u \bm{\delta}$ if and only if the closed-loop maps lie in the affine subspace defined by
\begin{align}
    \label{eq:sls_affine_subspace_constraints}
    (\mathbf{I} - \mathbf{Z} \mathbf{A}) \bm{\Phi}_x - \mathbf{Z} \mathbf{B} \bm{\Phi}_u = \mathbf{I}\,;
\end{align}
pairs of closed-loop responses $\{\bm{\Phi}_x, \bm{\Phi}_u\}$ that satisfy \eqref{eq:sls_affine_subspace_constraints} are said to be achievable \cite{wang2019system}.

Based on \eqref{eq:sls_affine_subspace_constraints}, many optimal control problems of practical interest, including classical $\mathcal{H}_2$ and $\mathcal{H}_\infty$ control synthesis \cite{hassibi1999indefinite}, can be equivalently posed as an optimization over the convex set of system responses.
Thanks to convexity, the optimal closed-loop responses $\{\bm{\Phi}^\star_x, \bm{\Phi}^\star_u\}$ can be computed efficiently, and the corresponding optimal control policy in the form of \eqref{eq:control_policy} can then be recovered by $\mathbf{K}^\star = \bm{\Phi}^\star_u{\bm{\Phi}^\star_x}^{-1}$. In particular, note that the structure that causality imposes on the closed-loop responses ensures that $\bm{\Phi}_x$ is always invertible.

In light of this correspondence, throughout the paper we will interchangeably identify a control policy $\bm{\pi}_T$ through the associated feedback gain matrix  $\mathbf{K}$ or through the closed-loop responses $\{\bm{\Phi}_x, \bm{\Phi}_u\}$ that it induces.

\subsection{The Clairvoyant Optimal Policy}
In the full-information setting we consider, there exists a unique noncausal control law that outperforms any other controller for every disturbance realization – the clairvoyant optimal policy $\bm{\pi}^c_T$. We formalize this statement in the following lemma, which we prove in the appendix by adapting the derivations of \cite{hassibi1999indefinite, martin2022safe}. For compactness, we define $\mathbf{Q} = \operatorname{blkdiag}(I_T \otimes Q, 0)$ and $\mathbf{R} = I_{T} \otimes R$.
\begin{lemma}
    For every disturbance realization $\bm{\delta}$, the clairvoyant optimal policy in \eqref{eq:clairvoyant_argmin_definition} is given by:
    \label{le:clairvoyant_optimal_policy}
    \begin{equation}
        \label{eq:clairvoyant_controller_definition}
        \bm{\pi}^c_T(\bm{\delta}) = -(\mathbf{R} + \mathbf{F}^\top\mathbf{Q} \mathbf{F})^{-1}\mathbf{F}^\top\mathbf{Q}\mathbf{G}\bm{\delta}\,,
    \end{equation}
    where $\mathbf{F} = (\mathbf{I} - \mathbf{Z} \mathbf{A})^{-1} \mathbf{Z} \mathbf{B}$ and $\mathbf{G} = (\mathbf{I} - \mathbf{Z} \mathbf{A})^{-1}$. Moreover, the cost that the clairvoyant optimal policy incurs is given by:
    \begin{equation}
        \label{eq:clairvoyant_controller_cost_incurred}
        J_T(\bm{\pi}^c_T, \bm{\delta}) = \bm{\delta}^\top \underbrace{\mathbf{G}^\top \mathbf{Q}(\mathbf{I} + \mathbf{F}\mathbf{R}^{-1}\mathbf{F}^\top\mathbf{Q})^{-1} \mathbf{G}}_{= \mathbf{C}_T} \bm{\delta}\,.%
    \end{equation}
\end{lemma}

Interestingly, the globally optimal clairvoyant control actions can be computed as a linear combination of past, present, and future disturbances \cite{hassibi1999indefinite}. As we have shown in \cite[Lem. 2]{martin2022safe}, the optimal control sequence in hindsight \eqref{eq:clairvoyant_controller_definition} can be equivalently computed by solving a quadratic SLS problem, i.e., an $\mathcal{H}_2$ problem, that imposes no causal constraints on the structure of closed-loop responses $\{\bm{\Phi}^c_x, \bm{\Phi}^c_u\}$ associated with the clairvoyant policy $\bm{\pi}^c_T$, i.e.,
\begin{align}
    \hspace{-.25mm}
    \{\bm{\Phi}^c_x, \bm{\Phi}^c_u\} = &~\argmin_{\bm{\Phi}_x, \bm{\Phi}_u} ~ \norm{
    \begin{bmatrix}
        \mathbf{Q}^{\frac{1}{2}} & 0\\
        0 & \mathbf{R}^{\frac{1}{2}}
    \end{bmatrix}
    \begin{bmatrix}
        \bm{\Phi}_x\\\bm{\Phi}_u
    \end{bmatrix}
    }_{F}^2 \label{eq:clairvoyant_policy_sls}\\
    &\st~
    \begin{bmatrix}
        \mathbf{I} - \mathbf{Z} \mathbf{A} & - \mathbf{Z} \mathbf{B}
    \end{bmatrix}
    \begin{bmatrix}
        \bm{\Phi}_x\\\bm{\Phi}_u
    \end{bmatrix} = \mathbf{I} \nonumber\,.
\end{align}
Note that, besides allowing for efficient computation of optimal controllers, \eqref{eq:clairvoyant_policy_sls} can be used to include additional convex requirements, such as safety structural constraints, in the definition of the clairvoyant policy \cite{martin2022safe}.

\section{Minimal Regret in finite horizon Control}
\label{sec:finite_horizon}
Towards establishing guarantees for regret-based receding horizon control, we first enable a tractable formulation of the optimal control problem that we will solve online. %
Building on our previous work \cite{martin2022safe}, we exploit the system level parametrization of achievable closed-loop responses to characterize a safe regret-optimal policy at performance level $\gamma$ as the solution of a convex optimization problem.
To do so, we consider the following objective function for $(\bm{\Phi}_x, \bm{\Phi}_u)$ satisfying \eqref{eq:safety_constraints} and \eqref{eq:sls_affine_subspace_constraints}:
\begin{equation}
    \label{eq:objective_maximization}
    \mathtt{Reg}_T^\star(\bm{\Phi}_x, \bm{\Phi}_u, \gamma, x_0) =
    \max_{\mathbf{w} \in \mathbfcal{W}} ~ \mathtt{Reg}_T(\bm{\Phi}_x, \bm{\Phi}_u, \gamma, \bm{\delta})\,,
\end{equation}
where,  with slight abuse of notation, we rewrite the regret $\mathtt{Reg}_T(\bm{\Phi}_x, \bm{\Phi}_u, \gamma, \bm{\delta})$ defined in \eqref{eq:regret_definition} as a function of the closed-loop system responses as
\begin{align}
    \mathtt{Reg}_T(\cdot) %
    = \norm{\bar{\mathbf{S}}^{\frac{1}{2}} \begin{bmatrix}
        \bm{\Phi}_x \\ \bm{\Phi}_u
    \end{bmatrix} \bm{\delta}}_2^2  \hspace{-.15cm} - \hspace{-.01cm}
    \norm{\mathbf{S}^{\frac{1}{2}} \begin{bmatrix}
        \bm{\Phi}_x^c \\ \bm{\Phi}_u^c
    \end{bmatrix} \bm{\delta}}_2^2 \hspace{-.15cm} -\gamma^2\norm{\mathbf{w}}_2^2\,,\label{eq:objective_negatively_weighted}
\end{align}
with $\mathbf{S} = \operatorname{blkdiag}(\mathbf{Q}, \mathbf{R})$, $\bar{\mathbf{Q}} = \operatorname{blkdiag}(I_T \otimes Q, P)$, and $\bar{\mathbf{S}} = \operatorname{blkdiag}(\bar{\mathbf{Q}}, \mathbf{R})$. In particular, note that \eqref{eq:objective_negatively_weighted} follows by combining \eqref{eq:closed_loop_behavior_x_K} and \eqref{eq:closed_loop_behavior_u_K} with \eqref{eq:lqr_cost}.
Further, observe that \eqref{eq:objective_negatively_weighted} is negatively weighted in the disturbance energy, as common in the classical literature on finite horizon $\mathcal{H}_\infty$ control \cite{james1995robust, green2012linear}. Crucially, however, the regret-optimal framework also introduces system-dependent optimal performance weights by subtracting the cost incurred by the clairvoyant optimal policy $\bm{\pi}^c_T$. By maximizing over $\mathbf{w}$ in \eqref{eq:objective_maximization}, we then penalize the worst-case loss relative to $\bm{\pi}^c_T$. Intuitively, this metric forces our controller to closely track the performance of $\bm{\pi}^c_T$ whenever a low cost can be attained, while bearing a higher cost whenever the clairvoyant optimal policy $\bm{\pi}^c_T$ does so. Loosely speaking, minimizing regret induces adaptivity in the controllers by letting them know what disturbances are worth spending more control effort on.

We now establish desirable properties of the objective function in \eqref{eq:objective_maximization}. Then, we provide conditions that make the computation of a minimax control policy amenable to standard techniques in convex optimization. For ease of exposition, we defer the discussion on the numerical implementation to Section~\ref{sec:numerical}. We first show that the worst-case dynamic regret objective $\mathtt{Reg}_T^\star(\bm{\Phi}_x, \bm{\Phi}_u, \gamma, x_0)$ is convex in the closed-loop system responses; the proof is detailed in Appendix~\ref{app:proof_convexity_lower_bound}.
\begin{proposition}
    \label{prop:convexity_lower_bound}
    For any fixed performance level $\gamma \in \mathbb{R}_{\geq 0}$, the function $(\bm{\Phi}_x, \bm{\Phi}_u) \mapsto \mathtt{Reg}_T^\star(\bm{\Phi}_x, \bm{\Phi}_u, \gamma, x_0)$ is convex, lower semicontinuous, proper, and bounded below by zero.
\end{proposition}

An admissible control policy that minimizes \eqref{eq:objective_maximization} can be efficiently computed if the function $\mathbf{w} \mapsto \mathtt{Reg}_T(\bm{\Phi}_x, \bm{\Phi}_u, \gamma, \bm{\delta})$ is concave. In this case, the optimization problem in \eqref{eq:objective_maximization} simply calls for the maximization of a concave function over the convex set $\mathbfcal{W}$, and a minimax control policy can be expressed as the solution of a tractable convex-concave optimization problem. Let $\bm{\Phi}_0 = \operatorname{col}(\bm{\Phi}^0_x, \bm{\Phi}^0_u)$ and $\bm{\Phi}_w = \operatorname{col}(\bm{\Phi}^w_x, \bm{\Phi}^w_u)$.
With this notation in place, we equivalently rewrite \eqref{eq:objective_negatively_weighted} as
\begin{align*}
    \mathtt{Reg}_T(\bm{\Phi}_x, \bm{\Phi}_u, \gamma, \bm{\delta}) &= x_0^\top
    \left(
    \bm{\Phi}_0^\top
    \bar{\mathbf{S}}
    \bm{\Phi}_0 -
    \mathbf{C}_{T}^{[0, 0]}
    \right)
    x_0 \\
    & + 2 x_0^\top
    \left(
    \bm{\Phi}_0^\top
    \bar{\mathbf{S}}
    \bm{\Phi}_w -
    \mathbf{C}_{T}^{[0, 1:T]}
    \right)
    \mathbf{w} \\
    & + \mathbf{w}^\top
    \left(
    \bm{\Phi}_w^\top
    \bar{\mathbf{S}}
    \bm{\Phi}_w -
    \mathbf{C}_{T}^{[1:T, 1:T]}
    - \gamma^2 \mathbf{I}
    \right)
    \mathbf{w}\,,
\end{align*}
where $\mathbf{C}_T$ is as in \eqref{eq:clairvoyant_controller_cost_incurred}.
Therefore, for a given performance level $\gamma \in \mathbb{R}_{\geq 0}$, the desired concavity condition is met if there exists a pair of achievable closed-loop responses $\bm{\Phi}_x$ and $\bm{\Phi}_u$ that satisfy the quadratic matrix inequality
\begin{equation}
    \label{eq:regret_objective_qmi}
    \mathbf{R}_T^\gamma(\bm{\pi}_T) =
    \gamma^2 \mathbf{I}
    +
    \mathbf{C}_{T}^{[1:T, 1:T]}
    -
    \bm{\Phi}_w^\top
    \bar{\mathbf{S}}
    \bm{\Phi}_w
    \succeq 0\,.
\end{equation}
Equation \eqref{eq:regret_objective_qmi} sets a lower bound on the best achievable performance level $\gamma$ that we can attain while maintaining tractability. Note that \eqref{eq:regret_objective_qmi} is necessary and sufficient for concavity if the perturbation $\mathbf{w}$ is an arbitrary disturbance sequence with bounded energy, compare also \cite[Th. 3]{martin2022safe}. Instead, even though \eqref{eq:regret_objective_qmi} is not necessary for particular choices of  $\mathbfcal{W}$ as those considered in \cite{didier2022system}, in these cases only upper and lower bounds on the minimum performance level are currently available.
Finally, we also note that \eqref{eq:regret_objective_qmi} guarantees the existence of a saddle point solution for the considered minimax problem \cite[Th. 2.3]{bacsar2008h}.

Applying the Schur complement to \eqref{eq:regret_objective_qmi} leads to:
\begin{equation}
    \label{eq:regret_lmi_maximization_concave}
    \begin{bmatrix}
        \gamma^2 \mathbf{I} + \mathbf{C}_{T}^{[1:T, 1:T]} & \bm{\Phi}_w^\top
        \bar{\mathbf{S}}^\frac{1}{2}\\
        \bar{\mathbf{S}}^\frac{1}{2}
        \bm{\Phi}_w
        & \mathbf{I}
    \end{bmatrix} \succeq 0\,.
\end{equation}
For a given initial condition $x_0 \in \mathbb{R}^n$ and planning horizon $T \in \mathbb{N}$, we denote the set of all admissible control policies that satisfy the linear matrix inequality \eqref{eq:regret_lmi_maximization_concave} with performance level $\gamma$ by $\bm{\Pi}_T^\gamma(x_0) \subseteq \bm{\Pi}_T(x_0)$, and the set of initial conditions for which one such policy exists by
\begin{equation}
    \label{eq:X_T_gamma_definition}
    \mathcal{X}_T^\gamma = \{x \in \mathbb{R}^n : \bm{\Pi}_T^\gamma(x) \neq \emptyset\} \subseteq \mathcal{X}_T\,.
\end{equation}
In particular, we remark that, while $x_0$ does not affect \eqref{eq:regret_lmi_maximization_concave} directly, the value of the initial state plays a role in the definition of the set $\bm{\Pi}_T^\gamma(x_0)$ due to the presence of the mixed state and input constraints \eqref{eq:safety_constraints}. Further, as at every state we are interested in satisfying \eqref{eq:dynamic_regret_bound} with as small $\gamma$ as possible, we define the minimum gain function $\gamma_T^\star : \mathcal{X}_T \to \mathbb{R}_{\geq 0}$ as
\begin{equation}
\label{eq:min_gain_gamma}
  \gamma^\star_T(x) = \min_{\gamma} \{\gamma : \bm{\Pi}_T^\gamma(x) \neq \emptyset\}\,.
\end{equation}
Then, for any $x_0 \in \mathcal{X}_T$ and any $\gamma \geq \gamma^\star_T(x_0)$, we consider the following tractable optimization problem:\footnote{Reformulations that enable efficient implementation via convex optimization techniques are discussed in Section~\ref{sec:numerical}.}
\begin{alignat}{3}
    V_T^\star(x_0, \gamma) = & && \hspace{-1.5cm} \min_{\bm{\Phi}_x,\bm{\Phi}_u} ~ \max_{\mathbf{w} \in \mathbfcal{W}} ~ \mathtt{Reg}_T(\bm{\Phi}_x, \bm{\Phi}_u, \gamma, \bm{\delta}) \label{eq:finite_horizon_value_function}\\
        & \st && ~ \eqref{eq:sls_affine_subspace_constraints}\,, \eqref{eq:regret_objective_qmi}\,, \nonumber\\ %
        & && (\mathbf{x}, \mathbf{u}) \in \mathbfcal{Z}\,, ~ \forall \mathbf{w} \in \mathbfcal{W}\,, \nonumber\\
        & && \bm{\Phi}_x, \bm{\Phi}_u \text{ with causal sparsities} \nonumber \,,%
\end{alignat}
where $\mathbfcal{Z}$ compactly denote the set of admissible state and input signals that satisfy \eqref{eq:safety_constraints} along the entire planning horizon. By design, the minimax control policy computed as a solution to \eqref{eq:finite_horizon_value_function} complies with the safety constraints \eqref{eq:safety_constraints} at all times, and drives the system to the terminal set $\mathcal{Z}_f$ in face of the uncertain disturbance realizations. Further, the dynamic regret bound \eqref{eq:dynamic_regret_bound} holds with $\beta(x_0, \gamma)$ equal to the value function $V_T^\star(x_0, \gamma)$.

\section{Receding Horizon Regret-Optimal Control}
\label{sec:infinite_horizon}
In this section, we turn our attention to the infinite horizon control problem and present our main contributions. We show that, if the terminal ingredients $\mathcal{Z}_f$ and $V_f(x_T)$ are appropriately chosen, the receding horizon control law $\bm{\mu}_T^s$ that implements only the first $s \in \mathbb{I}_{[1, T]}$ control actions of a minimax policy in \eqref{eq:finite_horizon_value_function} is recursively feasible as per Theorem~\ref{th:recursive_feasibility_terminal_rpi} and $\ell_2$-stable. In particular, our method enjoys finite regret with respect to the infinite horizon clairvoyant optimal policy $\bm{\pi}^c_\infty$ as per Theorem~\ref{th:rhc_finite_gain_infinite_horizon}. In other words, if \eqref{eq:finite_horizon_value_function} is feasible for $(x_0, \gamma)$, then the closed-loop system is guaranteed to satisfy the constraints at all times, and there exists $\overline{\gamma} \in \mathbb{R}_{\geq 0}$ and $\overline{\beta} \in \mathbb{R}_{\geq 0}$ such that the bound
\begin{align}
    \label{eq:infinite_horizon_dynamic_regret_bound}
    \hspace{-.2cm}
    J_\infty(\bm{\mu}_T^s, \mathbf{w}_{\infty})
    -
    J_\infty(\bm{\pi}_{\infty}^c, \mathbf{w}_{\infty})
    -
    \overline{\gamma}^2 \norm{\mathbf{w}_{\infty}}_2^2 \leq \overline{\beta}\,,
\end{align}
holds for any infinite disturbance sequence $\mathbf{w}_{\infty}$ with realizations $w_t \in \mathcal{W}$. We remark that, unlike in the $\mathcal{H}_\infty$ control case \eqref{eq:h_infinity_bound},
the disturbance sequence $\mathbf{w}_\infty$ directly affects the performance evaluation through the term $J_\infty(\bm{\pi}_{\infty}^c, \mathbf{w}_{\infty})$ in \eqref{eq:infinite_horizon_dynamic_regret_bound} – thus teaching the policy how much effort to spend in trying to counteract each individual disturbance sequence $\mathbf{w}_\infty$.

To streamline the presentation of our results, we consider the standard implementation of MPC that corresponds to choosing $s = 1$, see also \cite{rawlings2017model, mayne2000constrained}; analogous results hold for a general $s \in \mathbb{I}_T$. Hence, we formally define our receding horizon control law $\bm{\mu}_T^1 = \mu_T : \mathcal{X}_T \times \mathbb{R}_{\geq 0} \to \mathbb{R}^m$ as:
\begin{equation}
    \label{eq:receding_horizon_control_law}
    \mu_T(x_0, \gamma) = u_0^\star(x_0, \gamma) = {\bm{\Phi}_u^{\star}}^{[0,0]} x_0\,,
\end{equation}
so that the closed-loop system behavior can be described by
\begin{equation}
    \label{eq:closed_loop_system_receding_horizon_law}
    x_{t+1} = A x_t + B \mu_T(x_t, \gamma) + w_t\,.
\end{equation}
We make the following assumptions on the terminal constraint set $\mathcal{Z}_f$ and terminal penalty $V_f(x_T) = x_T^\top P x_T$. Note that these are adaptations of standard assumptions in the robust MPC literature \cite{rawlings2017model, mayne2000constrained}.
\begin{assumption}
    \label{ass:stab_unit_circle_obs}
    $(A, B)$ is stabilizable and $(A, Q^{\frac{1}{2}})$ is observable on the unit circle \cite{hassibi1999indefinite}.
\end{assumption}
\begin{assumption}
    \label{ass:terminal_cost}
    The terminal cost $P \succ 0$ satisfies the sign indefinite algebraic Riccati equation of an unconstrained linear $\mathcal{H}_\infty$ state-feedback control problem at performance level $\gamma_f$ \cite{hassibi1999indefinite}:
\begin{align}
    \label{eq:sign_indefinite_dare}
    P &= Q + A^\top \bar{P} A - A^\top \bar{P} B (R + B^\top \bar{P} B)^{-1} B^\top \bar{P} A\,,
\end{align}
where $\bar{P} = P + P (\gamma_f^2 I - P)^{-1} P$ and $\gamma_f^2 I - P \succ 0$.
\end{assumption}
Now, let $K_f = -(R + B^\top \bar{P} B)^{-1} B^\top P A$ denote the \textit{auxiliary stabilizing state feedback gain} derived from the solution of the above $\mathcal{H}_\infty$ problem. %
\begin{assumption}
    \label{ass:terminal_set}
    The terminal set $\mathcal{Z}_f$ is constraint admissible, $\mathcal{Z}_f \subseteq \{x : (x, K_f x) \in \mathcal{Z}\}$ and robust positively invariant under the local control law $u_t = K_f x_t$,
    \begin{equation*}
        (A + BK_f)x_t + w_t \in \mathcal{Z}_f\,, ~\forall x_t \in \mathcal{Z}_f\,, ~\forall w_t \in \mathcal{W}\,.
    \end{equation*}
\end{assumption}

Unlike classical $\mathcal{H}_2$ and $\mathcal{H}_\infty$ controllers, evaluating the infinite horizon regret value function requires aggregate information about the entire history of disturbance realizations. Hence, showing that the regret-to-go decreases along the trajectories of the closed-loop system \eqref{eq:closed_loop_system_receding_horizon_law} and can serve as a Lyapunov function is a challenging task. To get around this problem, Assumption~\ref{ass:terminal_cost} allows us to compute a quadratic upper bound on the infinite horizon value function that only depends on the current state. %
Once the system has reached the terminal safe set, Assumption~\ref{ass:terminal_set} guarantees that $u_t = K_f x_t$ is constraint admissible. Moreover, being $K_f$ optimal in the $\mathcal{H}_\infty$ sense, the optimal regret cannot be higher than the worst-case cost incurred by this policy. As we will show, this choice ensures that the finite horizon regret-optimal computation is recursively feasible and stabilizing, at the cost of some conservatism in our theoretical performance bound.%
 \subsection{Monotonicity Properties of the Clairvoyant Optimal Policy}
To prove recursive feasibility of the receding horizon control law \eqref{eq:receding_horizon_control_law}, we need to ensure that \eqref{eq:regret_objective_qmi} recursively holds when moving from time $T$ to $T+1$. To do so, we first study how the cost incurred by $\bm{\pi}_T^c$ relates to that of $\bm{\pi}_{T+1}^c$ – that is, the clairvoyant optimal policy over a horizon of $T+1$ – when focusing on subsequences of length $T$. %
\begin{lemma}
    \label{le:relations_clairvoyant}
    Let $\mathbf{C}_T$ and $\mathbf{C}_{T+1}$ denote the matrices defining the cost of the clairvoyant optimal policies according to \eqref{eq:clairvoyant_controller_cost_incurred} for a planning horizon of length $T$ and $T + 1$, respectively.
    Then, the following chain of matrix inequalities holds:
    \begin{equation}
        \label{eq:clairvoyant_inequalities_chain}
        \mathbf{C}_{T+1}^{[1:T+1, 1:T+1]} \preceq \mathbf{C}_{T} \preceq \mathbf{C}_{T+1}^{[0:T, 0:T]}\,.
    \end{equation}
\end{lemma}

The proof of Lemma~\ref{le:relations_clairvoyant} can be found in Appendix~\ref{app:proof_relations_clairvoyant}. This result states that, for any $\bm{\delta}$, the cost \eqref{eq:clairvoyant_controller_cost_incurred} incurred by $\bm{\pi}_T^c$ is bounded between the cost that $\bm{\pi}_{T+1}^c$ incurs on the extended disturbance sequences $\operatorname{col}(0, \bm{\delta})$ and $\operatorname{col}(\bm{\delta}, 0)$. We remark that non-causality and optimality for each disturbance sequence – properties that are unique to clairvoyant policies –  are crucial to establishing \eqref{eq:clairvoyant_inequalities_chain}.
Indeed, when considering causal policies defined over different horizons, monotonicity properties as those in \eqref{eq:clairvoyant_inequalities_chain} can only be derived with respect to a single performance measure, such as the worst-case control cost \cite{goulart2009control}. Instead, \eqref{eq:clairvoyant_inequalities_chain} indicate stronger connections, which allow recovering monotonicity properties also known for causal policies as simple corollaries.

\begin{corollary}
    \label{cor:eigenvalues_clairvoyant}
    The clairvoyant worst-case cost function $c_{\operatorname{max}} : \mathbb{N} \to \mathbb{R}_{\geq 0}$ defined as $c_{\operatorname{max}}(T) = \lambda_{\operatorname{max}}(\mathbf{C}_T)$ is monotonically non-decreasing with respect to the planning horizon $T$. Similarly, the clairvoyant best-case cost function $c_{\operatorname{min}} : \mathbb{N} \to \mathbb{R}_{\geq 0}$ defined as $c_{\operatorname{min}}(T) = \lambda_{\operatorname{min}}(\mathbf{C}_T)$ is monotonically non-increasing with respect to the planning horizon $T$.
\end{corollary}

A proof of Corollary~\ref{cor:eigenvalues_clairvoyant} is given in Appendix~\ref{app:proof_eigenvalues_clairvoyant} for completeness. The inequalities presented in this section are key to our subsequent derivations, as, by induction, they allow accounting for the mismatch between the finite horizon clairvoyant policy $\bm{\pi}^c_T$, which we repeatedly adopt as control benchmark during online operation, and the infinite horizon clairvoyant policy $\bm{\pi}_\infty^c$ we compare against in \eqref{eq:infinite_horizon_dynamic_regret_bound}.

\subsection{Recursive Feasibility Properties}
Inspired by the proof philosophy in \cite{goulart2009control} for receding horizon $\mathcal{H}_\infty$ control and leveraging Lemma~\ref{le:relations_clairvoyant} to bound the effect of introducing a clairvoyant control benchmark via regret minimization, we proceed to show that, under Assumptions~\ref{ass:stab_unit_circle_obs}-\ref{ass:terminal_set}, the optimization problem \eqref{eq:finite_horizon_value_function} for the closed-loop system \eqref{eq:closed_loop_system_receding_horizon_law} can be solved online at all times – assuming that it is feasible initially. %
To do so, we first prove that, for any $\gamma \geq \gamma_f$, the sets $\mathcal{X}_T^\gamma$ defined in \eqref{eq:X_T_gamma_definition} %
are monotonically non-decreasing in $T$ with respect to the set inclusion. To simplify the notation, we define:%
\begin{align}
    \label{eq:h_infinity_objective_qmi}
    \mathbf{H}_T^\gamma(\bm{\pi}_T) &=  \gamma^2 \mathbf{I} -
    \bm{\Phi}_w^\top
    \bar{\mathbf{S}}
    \bm{\Phi}_w
    \succeq 0\,.
\end{align}
This quadratic matrix inequality is reminiscent of \eqref{eq:regret_objective_qmi} and guarantees concavity of the  $\mathcal{H}_\infty$ objective \eqref{eq:h_infinity_bound}. %
Additionally, we observe that, by definition, it holds that:
\begin{equation*}
    \mathbf{R}_T^\gamma(\pi) = \mathbf{H}_T^\gamma(\pi) + \mathbf{C}_{T}^{[1:T, 1:T]}\,.
\end{equation*}

\begin{proposition}
    \label{prop:chain_of_inclusions}
    Let Assumptions~\ref{ass:stab_unit_circle_obs}-\ref{ass:terminal_set} hold with performance level $\gamma_f \leq \gamma$. Then, the following set inclusion property holds:
    \begin{equation*}
        \mathcal{Z}_f \subseteq \mathcal{X}_1^\gamma \subseteq \cdots \subseteq \mathcal{X}_{T-1}^\gamma \subseteq \mathcal{X}_T^\gamma \subseteq \mathcal{X}_{T+1}^\gamma \subseteq \cdots\,.
    \end{equation*}
\end{proposition}

This chain of set inclusions, which we prove in Appendix~\ref{app:proof_chain_of_inclusions}, implies that, for any $x \in \mathcal{X}_T$, the minimum gain function $\gamma^\star_T(\cdot)$ defined in \eqref{eq:min_gain_gamma} satisfies the inequality $\gamma^\star_{T+1}(x) \leq \operatorname{max}\{\gamma_f, \gamma^\star_T(x)\}$.
We continue by showing that the proposed receding horizon scheme is recursively feasible; we defer the proof to Appendix~\ref{app:proof_recursive_feasibility_terminal_rpi}.

\begin{theorem}
    \label{th:recursive_feasibility_terminal_rpi}
    Let Assumptions~\ref{ass:stab_unit_circle_obs}-\ref{ass:terminal_set} hold. Then, for every performance level $\gamma \geq \operatorname{max}\{\gamma_T^\star(x_0), \gamma_f\}$ and every $T \in \mathbb{N}$, the set $\mathcal{X}_T^\gamma$ is robust positively invariant under the receding horizon control law \eqref{eq:receding_horizon_control_law}, i.e., $A x + B\mu_T(x, \gamma) + w \in \mathcal{X}_T^\gamma$ for all $x \in \mathcal{X}_T^\gamma$ and all $w \in \mathcal{W}$.
\end{theorem}
When implementing the receding horizon control law \eqref{eq:receding_horizon_control_law}, each time a new measure of the state $x_t$ is made available, the optimal performance level $\gamma_T^\star(x_t)$ can be computed by bisection, see also the discussion after Proposition~\ref{prop:implementation_structured}. In particular, Theorem~\ref{th:recursive_feasibility_terminal_rpi} ensures that $\gamma_T^\star(x_t) \leq \operatorname{max}(\gamma_T^\star(x_0), \gamma_f)$ for all $t \in \mathbb{N}$.

\subsection{Bounded Regret in Receding Horizon Control}
Building on the recursive feasibility properties established in the previous section, we now show that the proposed receding horizon control scheme with terminal ingredients derived from an auxiliary $\mathcal{H}_\infty$ policy is safe, in the sense of $(x_t, u_t) \in \mathcal{Z}$ at all times, and achieves finite regret with respect to the infinite horizon clairvoyant optimal policy $\bm{\pi}_\infty^c$.
\begin{theorem}
    \label{th:rhc_finite_gain_infinite_horizon}
    Let Assumptions~\ref{ass:stab_unit_circle_obs}-\ref{ass:terminal_set} hold.
    Then, for every performance level $\gamma \geq \operatorname{max}(\gamma_T^\star(x_0), \gamma_f)$, the closed-loop system gain from the disturbance energy to the dynamic regret in \eqref{eq:infinite_horizon_dynamic_regret_bound} is bounded above by $\gamma$. Moreover, for any $x_0 \in \mathcal{X}^\gamma_T$, the closed-loop system satisfies $(x_t, u_t) \in \mathcal{Z}$ at all times.
\end{theorem}

The proof of Theorem~\ref{th:rhc_finite_gain_infinite_horizon} is presented in Appendix~\ref{app:proof_rhc_finite_gain_infinite_horizon}. We remark that, for any $\mathbf{w}_\infty$ such that $\normtwo{\mathbf{w}_\infty}<\infty$, bounded dynamic regret as per \eqref{eq:infinite_horizon_dynamic_regret_bound} implies that the realized control cost $J_\infty(\mu_T, \mathbf{w}_\infty)$ is also bounded. To see this, it is sufficient to note that the $\mathcal{H}_\infty$ controller $K_f$ guarantees bounded cost $J_\infty(K_f, \mathbf{w}_{\infty})$, and that $J_\infty(\bm{\pi}_\infty^c, \mathbf{w}_{\infty}) \leq J_\infty(K_f, \mathbf{w}_{\infty})$ by definition.

We now turn our attention to stability of the closed-loop system \eqref{eq:closed_loop_system_receding_horizon_law} as a Corollary of Theorem~\ref{th:rhc_finite_gain_infinite_horizon}.%

\begin{corollary}
    Let Assumptions~\ref{ass:stab_unit_circle_obs}-\ref{ass:terminal_set} hold with performance level $\gamma \geq \operatorname{max}(\gamma_T^\star(x_0), \gamma_f)$. Then, the origin of the undisturbed closed-loop system $x_{t+1} = Ax_t + B\mu_T(x_t, \gamma)$ is locally exponentially stable with region of attraction $\mathcal{X}_T^\gamma$. Moreover, if $Q\succ 0$, the closed-loop system \eqref{eq:closed_loop_system_receding_horizon_law} exhibits a finite $\ell_2$ gain also from disturbances to state and input trajectories, that is:
    \begin{equation}
    \label{eq:l2_stability_input_state}
    \max_{\normtwo{\mathbf{w}_\infty} < \infty} ~ \frac{\normtwo{\mathbf{x}_\infty}}{\normtwo{\mathbf{w}_\infty}}< \infty,~ \max_{\normtwo{\mathbf{w}_\infty} < \infty} ~ \frac{\normtwo{\mathbf{u}_\infty}}{\normtwo{\mathbf{w}_\infty}} < \infty\,.
\end{equation}
\end{corollary}

\begin{proof}
    Recalling from the proof of Theorem~\ref{th:rhc_finite_gain_infinite_horizon} that $V_f(\cdot)$ is a robust control Lyapunov function in a neighborhood of the origin, the result for the undisturbed system follows by resorting to classical results in MPC, see, for instance, \cite[Sect. 3]{mayne2000constrained}. Finally, since $J_\infty(\mu_T, \mathbf{w}_\infty) \geq \lambda_{\operatorname{min}}(Q) \normtwo{\mathbf{x}_\infty}^2 + \lambda_{\operatorname{min}}(R) \normtwo{\mathbf{u}_\infty}^2$ and both $\lambda_{\operatorname{min}}(Q)$ and $\lambda_{\operatorname{min}}(R)$ are strictly positive, we conclude that $\normtwo{\mathbf{x}_\infty}^2$ and $\normtwo{\mathbf{u}_\infty}^2$ are themselves bounded. Hence, \eqref{eq:l2_stability_input_state} follows.
\end{proof}

The interplay between regret and classical definitions of stability has been recently studied in \cite{karapetyan2022implications} in the context of linear dynamical systems subject to adversarial disturbances, and in \cite{nonhoff2022relation} for deterministic nonlinear systems. Our analysis takes another step towards understanding the relationship between these notions by showing that implementing a regret-optimal control policy in a receding horizon fashion does not compromise the stability of the closed-loop system – despite the significantly higher performance that can be achieved when disturbances do not follow classical assumptions.

\section{Numerical Experiments}
\label{sec:numerical}
In this section, we first discuss how the minimax policy optimization problem \eqref{eq:finite_horizon_value_function} can be efficiently solved with convex programming techniques. Then, we present numerical results to show how the proposed receding horizon control scheme retains the ability of safe finite horizon regret-optimal policy to interpolate between, or even prevail over, the performance of constrained $\mathcal{H}_2$ and $\mathcal{H}_\infty$ controllers \cite{martin2022safe, didier2022system}. In particular, our numerical experiments showcase that regret-optimal policies can achieve superior control performance despite less frequent solution of the policy optimization problem during online operation.

\subsection{Numerically Efficient Implementation}
\label{subsec:numerical_implementation}
For simplicity, we assume that the disturbance set $\mathbfcal{W}$ is a polytope. Then, there exists a matrix $\mathbf{H}_w$ and a vector $\mathbf{h}_w$ such that $\mathbfcal{W} = \{\mathbf{w} : \mathbf{H}_w \mathbf{w} \leq \mathbf{h}_w\}$; we refer the interested reader to \cite[Ch. 6]{goulart2007affine} for insights on how to generalize our results to disturbance sets characterized by affine conic inequalities.
Similarly, we assume that $\mathbfcal{Z} = \{\mathbf{z} : \mathbf{H}_z \mathbf{z} \leq \mathbf{h}_z\}$.
The following proposition presents an equivalent convex reformulation of the optimal control problem of interest. To ease readability, we defer all proofs to the appendix.
\begin{proposition}
    \label{prop:implementation_structured}
    Consider the evolution of the linear time-invariant system \eqref{eq:lti_system_dynamics} over a planning horizon of length $T \in \mathbb{N}$, and the constrained regret-optimal control problem at performance level $\gamma$:
    \begin{subequations}
    \label{prob:safe_regret_optimal_original}
    \begin{alignat}{3}
    &~\min_{\bm{\Phi}_x,\bm{\Phi}_u} ~ && \hspace{-0.25cm} \mathtt{Reg}_T^\star(\bm{\Phi}_x, \bm{\Phi}_u, \gamma, x_0) \label{eq:safe_regret_optimal_problem_original}\\
        &\st ~&&\eqref{eq:sls_affine_subspace_constraints}\,, \eqref{eq:regret_objective_qmi}\,, \nonumber\\
        & &&
        \mathbf{H}_z
        \begin{bmatrix}
            \bm{\Phi}_x \\ \bm{\Phi}_u
        \end{bmatrix} \hspace{-.1cm}
        \begin{bmatrix}
            x_0 \\ \mathbf{w}
        \end{bmatrix}
        \leq \mathbf{h}_z, ~ \forall \mathbf{w} : \mathbf{H}_w \mathbf{w} \leq \mathbf{h}_w, \label{eq:safety_constraints_compact_Phi}\\
        & && \bm{\Phi}_x, \bm{\Phi}_u \text{ with causal sparsities}\label{eq:closed_loop_responses_causal_sparsities}\,.
    \end{alignat}
    \end{subequations}
    Then, \eqref{prob:safe_regret_optimal_original} is equivalently formulated as the following convex optimization problem:
    \begin{subequations}
    \label{prob:safe_regret_optimal_sdp}
    \begin{alignat}{3}
        &~\min_{\bm{\Phi}_x,\bm{\Phi}_u, \mathbf{Y}, \tau, \bm{\eta}} && \hspace{0.5cm} 2\mathbf{h}_w^\top \bm{\eta} + \tau \label{eq:objective_to_be_reformulated}\\
        &\st ~&& \eqref{eq:sls_affine_subspace_constraints}\,, ~ \eqref{eq:closed_loop_responses_causal_sparsities}\,, ~ \bm{\eta} \geq 0\,, ~ \mathbf{Y}_{ij} \geq 0\,,\nonumber\\
        & &&
        \mathbf{Y}^\top \mathbf{h}_w \leq \mathbf{h}_z - \mathbf{H}_z
        \bm{\Phi}_0
        x_0\,,\label{eq:sdp_dual_safety_constraints_1}\\
        & &&
        \mathbf{H}_z
        \bm{\Phi}_w
        = \mathbf{Y}^\top \mathbf{H}_w\,,\label{eq:sdp_dual_safety_constraints_2}\\
        & && \bm{\Theta}(\bm{\Phi}_x, \bm{\Phi}_u, \gamma, x_0)
         \succeq 0 \label{eq:sdp_schur_constraints}\,,%
    \end{alignat}
    where $\bm{\Theta}(\bm{\Phi}_x, \bm{\Phi}_u, \gamma, x_0)$ is defined as
    \begin{equation*}
        \begin{bmatrix}
            \tau + x_0^\top \mathbf{C}_T^{[0,0]} x_0 & \bm{\eta}^\top \mathbf{H}_w + x_0^\top \mathbf{C}_T^{[0, 1:T]} & x_0^\top \bm{\Phi}_0^\top \bar{\mathbf{S}}^\frac{1}{2}\\
            \mathbf{H}_w

^\top \bm{\eta} + \mathbf{C}_T^{[1:T, 0]} x_0 & \gamma^2 \mathbf{I} + \mathbf{C}_T^{[1:T, 1:T]} & \bm{\Phi}_w^\top \bar{\mathbf{S}}^\frac{1}{2}\\
            \bar{\mathbf{S}}^\frac{1}{2} \bm{\Phi}_0 x_0 &
            \bar{\mathbf{S}}^\frac{1}{2} \bm{\Phi}_w & \mathbf{I}
        \end{bmatrix}\,.
    \end{equation*}
\end{subequations}
\end{proposition}

Note that the semidefinite optimization problem \eqref{prob:safe_regret_optimal_sdp} can be efficiently solved using off-the-shelf interior point methods. Moreover, given an optimal solution $\{\bm{\Phi}_x^\star, \bm{\Phi}_u^\star\}$ to \eqref{prob:safe_regret_optimal_sdp}, the corresponding optimal state-feedback controller can then be reconstructed by $\mathbf{K}^\star = \bm{\Phi}_u^\star {\bm{\Phi}_x^\star}^{-1}$. Independently of how the disturbances are generated, this optimal safe control policy guarantees that the dynamic regret bound \eqref{eq:dynamic_regret_bound} holds with performance level $\gamma$. Furthermore, the receding horizon control law that results by repeatedly implementing the first $s \in \mathbb{I}_{[1, T]}$ optimal control actions is recursively feasible and satisfies \eqref{eq:infinite_horizon_dynamic_regret_bound} with performance level not greater than $\operatorname{max}(\gamma, \gamma_f)$.

To ensure that the regret-optimal policy synthesized at each planning stage through \eqref{prob:safe_regret_optimal_sdp} achieves the optimal performance level $\gamma_T^\star(x_t)$, one can search over $\gamma$ through bisection and solve several corresponding instances of \eqref{prob:safe_regret_optimal_sdp}. Clearly, this procedure comes with an increased computational cost. To mitigate this aspect, one can further assume that the energy of $\mathbf{w}$ is bounded by a constant $\sigma$ – while its specific realization remains arbitrary.  Using this assumption, our next proposition suggests a convex reformulation in the line of \cite{didier2022system} that jointly synthesizes a regret-optimal policy and an optimal performance level. In the interest of clarity, we point out that our main theoretical results hold irrespective of the bounded energy assumption on $\mathbf{w}$ over the interval $[0, T-1]$.
\begin{proposition}
    \label{prop:implementation_bounded_energy}
    Consider the evolution of the linear time-invariant system \eqref{eq:lti_system_dynamics} over a planning horizon of length $T \in \mathbb{N}$, and the constrained regret-optimal control problem:
    \begin{align}
    &~\min_{\bm{\Phi}_x,\bm{\Phi}_u} ~ \max_{\norm{\mathbf{w}}_2 \leq \sigma} ~ J_T(\bm{\pi}_T, \bm{\delta}) - J_T(\bm{\pi}^c_T, \bm{\delta}) + V_f(x_T)\label{prob:safe_regret_optimal_original_bounded_energy}\\
        &\st ~\eqref{eq:sls_affine_subspace_constraints}\,, \eqref{eq:regret_objective_qmi}\,, \eqref{eq:safety_constraints_compact_Phi}\,, \eqref{eq:closed_loop_responses_causal_sparsities}\,.\nonumber
    \end{align}
    Then, \eqref{prob:safe_regret_optimal_original_bounded_energy} is equivalent to the following convex optimization problem:
    \begin{subequations}
    \vspace{-.2cm}
    \label{prob:safe_regret_optimal_sdp_bounded_energy}
    \begin{align}
        &~\min_{\bm{\Phi}_x,\bm{\Phi}_u, \mathbf{Y}, \lambda, \gamma} ~ \gamma\\
        &\st~ \eqref{eq:sls_affine_subspace_constraints}\,, \eqref{eq:sdp_dual_safety_constraints_1}\,, \eqref{eq:sdp_dual_safety_constraints_2}\,, \eqref{eq:closed_loop_responses_causal_sparsities}\,, \nonumber \\
        & \qquad \lambda \geq 0\,, ~ \mathbf{Y}_{ij} \geq 0\,,\nonumber\\
        & \qquad \bm{\Xi}(\bm{\Phi}_x, \bm{\Phi}_u, \gamma, x_0)
         \succeq 0 \label{eq:sdp_schur_constraints_bounded_energy}\,,%
    \end{align}
    where $\bm{\Xi}(\bm{\Phi}_x, \bm{\Phi}_u, \gamma, x_0)$ is defined as
    \begin{equation*}
        \begin{bmatrix}
            x_0^\top \mathbf{C}_T^{[0,0]} x_0 + \gamma - \sigma \lambda & x_0^\top \mathbf{C}_T^{[0, 1:T]} & x_0^\top \bm{\Phi}_0^\top \bar{\mathbf{S}}^\frac{1}{2}\\
            \mathbf{C}_T^{[1:T, 0]} x_0 & \lambda \mathbf{I} + \mathbf{C}_T^{[1:T, 1:T]} & \bm{\Phi}_w^\top \bar{\mathbf{S}}^\frac{1}{2}\\
            \bar{\mathbf{S}}^\frac{1}{2} \bm{\Phi}_0 x_0 &
            \bar{\mathbf{S}}^\frac{1}{2} \bm{\Phi}_w & \mathbf{I}
        \end{bmatrix}\,.
    \end{equation*}
\end{subequations}
\end{proposition}

\subsection{Numerical Results}
\begin{table*}[hbt]
    \caption{Performance comparison: average control cost per unit of disturbance energy incurred by the $\mathcal{H}_2$, $\mathcal{H}_\infty$, and regret-optimal ($\mathcal{R}$) receding horizon safe control policy over 10 different disturbance realizations. For each disturbance profile, the minimum cost incurred by any of these policies is highlighted in \textcolor{mygreen}{green}.}
    \label{table:results}
    \vspace{3pt}
    \centering
    \hspace{-4.5cm} $\mathcal{H}_2$ \hspace{8.5cm} $\mathcal{H}_\infty$\vspace{1mm}\\
    \begin{tabular}{c|c}
        \hline
        $\mathbf{w}$ & $s = 1$\\
        \hline
        \rowcolor{lightgray} $\mathcal{N}(0,1)$ & {\color{mygreen} \textbf{8.42}}\\
        $\mathcal{U}_{[0.5, 1]}$ & 14.72\\
        \rowcolor{lightgray} $1$ & 18.56\\
        $\operatorname{ramp}$ & 12.54\\
        \rowcolor{lightgray} $\operatorname{sin}$ & 14.38\\
        $\operatorname{step+sin}$ & 10.19\\
        \rowcolor{lightgray} $\operatorname{sawtooth}$ & 12.15\\
        $\operatorname{stairs}$ & 15.63\\
        \hline
    \end{tabular}
    \hspace{5mm}
    \begin{tabular}{c|cccccccccc}
        \hline
        $\mathbf{w}$ & $s = 1$ & $s = 2$ & $s = 3$ & $s = 4$ & $s = 5$ & $s = 6$ & $s = 10$ & $s = 12$ & $s = 15$ & $s = 20$\\
        \hline
        \rowcolor{lightgray} $\mathcal{N}(0,1)$        &  9.02 & 10.07 & 10.73 & 10.62 & 10.65 & 10.88 & 10.93 & 11.08 & 11.07 & 11.17\\
        $\mathcal{U}_{[0.5, 1]}$  & 13.20 & 12.70 & 13.16 & 14.13 & 14.89 & 15.53 & 17.60 & 18.28 & 19.10 & 18.75\\
        \rowcolor{lightgray} $1$                       & 16.56 & 16.49 & 17.49 & 18.68 & 19.70 & 20.57 & 23.13 & 24.07 & 24.94 & 23.61\\
        $\operatorname{ramp}$     & 11.14 & 10.64 & 11.03 & 11.67 & 12.28 & 12.78 & 14.28 & 14.86 & 15.47 & 15.02\\
        \rowcolor{lightgray} $\operatorname{sin}$      & 13.38 & 12.88 & 13.34 & 14.58 & 15.14 & 15.87 & 19.55 & 17.55 & 17.50 & 16.99\\
        $\operatorname{step+sin}$ &  9.29 &  9.08 &  9.50 &  9.93 & 10.19 & 11.01 & 13.02 & 13.70 & 12.16 & 13.74\\
        \rowcolor{lightgray} $\operatorname{sawtooth}$ & 11.62 & 11.32 & 12.60 & 13.11 & 13.67 & 14.40 & 16.40 & 15.69 & 16.16 & 13.58\\
        $\operatorname{stairs}$   & 14.05 & 14.14 & 14.88 & 16.13 & 17.27 & 17.23 & 21.34 & 21.49 & 20.68 & 21.43\\
        \hline
    \end{tabular}
    \vspace{3mm}\\
    $\mathcal{R}$\vspace{1mm}\\
    \begin{tabular}{c|cccccccccc}
        \hline
        $\mathbf{w}$ & $s = 1$ & $s = 2$ & $s = 3$ & $s = 4$ & $s = 5$ & $s = 6$ & $s = 10$ & $s = 12$ & $s = 15$ & $s = 20$\\
        \hline
        \rowcolor{lightgray} $\mathcal{N}(0,1)$        & {\color{mygreen} \textbf{8.42}} & 11.30 & 12.19 & 12.47 & 12.38 & 12.65 & 12.39 & 12.47 & 12.38 & 12.55\\
        $\mathcal{U}_{[0.5, 1]}$  & 14.72 & 10.74 & 8.96 & 8.37 & 8.07 & 7.85 & 7.31 & 7.21 & {\color{mygreen} \textbf{7.10}} & 7.65\\
        \rowcolor{lightgray} $1$                       & 18.56 & 12.80 & 10.25 &  9.46 &  9.10 &  8.67 &  7.99 &  7.84 &  {\color{mygreen} \textbf{7.68}} &  8.42\\
        $\operatorname{ramp}$     & 12.54 &  9.00 &  7.39 &  6.83 &  6.59 &  6.33 &  5.96 &  5.86 &  {\color{mygreen} \textbf{5.75}} &  6.12\\
        \rowcolor{lightgray} $\operatorname{sin}$      & 14.39 & 10.87 &  9.37 &  8.98 &  8.85 &  8.62 &  {\color{mygreen} \textbf{7.20}} &  8.29 &  8.26 &  7.66\\
        $\operatorname{step+sin}$ & 10.19 &  7.86 &  6.87 &  6.50 &  6.43 &  6.31 &  {\color{mygreen} \textbf{5.71}} &  6.01 &  6.25 &  5.74\\
        \rowcolor{lightgray} $\operatorname{sawtooth}$ & 12.15 & 9.83 & 13.12 &  8.07 &  7.84 & 12.20 &  {\color{mygreen} \textbf{7.52}} & 11.18 & 11.21 &  7.65\\
        $\operatorname{stairs}$   & 15.63 & 11.04 &  9.40 &  8.38 &  8.00 &  8.25 &  {\color{mygreen} \textbf{7.14}} &  7.29 &  7.38 &  7.18\\
        \hline
    \end{tabular}
\end{table*}
We consider a linear time-invariant system \eqref{eq:lti_system_dynamics} with
\begin{equation*}
    A = \rho \begin{bmatrix}
        0.7 & 0.2 & 0\\
        0.3 & 0.7 & -0.1\\
        0 & -0.2 & 0.8
    \end{bmatrix}\,, ~~
    B = \begin{bmatrix}
        1 & 0.2\\
        2 & 0.3\\
        1.5 & 0.5
    \end{bmatrix}\,,
\end{equation*}
with spectral radius $\rho = 0.7$.\footnote{We choose an open-loop stable system to ease computation of terminal constrained sets using Algorithm~1 in \cite{rakovic2005invariant}. Note that our theoretical results hold independently of this choice.} We randomly select the initial condition $x_0 = \begin{bmatrix}
    -3.08 & 1.22 & -0.62
\end{bmatrix}^\top$. Further, we define the safe set $\mathcal{Z} = \{(x, u) \in \mathbb{R}^{5} : \norm{x}_\infty \leq 3.5\,, ~ \norm{u}_\infty \leq 2\}$, and we enforce robust constraint satisfaction for all disturbance realizations taking values in $\mathcal{W} = \{w \in \mathbb{R}^3 : -1 \leq \norm{w}_\infty \leq 1\}$. We consider a planning horizon of length $T = 20$ and a simulation of length $N = 60$. Last, we set the weight matrices to $Q = I_3$ and $R = I_2$ for all experiments.

Next, we synthesize different stabilizing receding horizon control policies following the $\mathcal{H}_2$, $\mathcal{H}_\infty$ and regret minimization philosophies; we discuss implementation details below.

\emph{1) Schemes based on $\mathcal{H}_2$ and $\mathcal{H}_\infty$ objectives:} We compute terminal cost matrices $P_2$ and $P_\infty$ for the expected and worst-case objectives, respectively, by solving the sign-definite algebraic Riccati equation
\begin{equation*}
    P_2 = Q + A^\top P_2 A - A^\top P_2 B (R + B^\top P_2 B)^{-1} B^\top P_2 A\,,
\end{equation*}
and the sign-indefinite algebraic Riccati equation \eqref{eq:sign_indefinite_dare}, where we iteratively determine the optimal value of $\gamma_f$ via bisection with a tolerance of $0.001$. From these two solutions, we derive auxiliary $\mathcal{H}_2$ and $\mathcal{H}_\infty$ optimal controllers as $K_2 = (R + B^\top P B)^{-1} B^\top P_2 A$ and $K_\infty = -(R + B^\top \bar{P}_\infty B)^{-1} B^\top P_\infty A$, where $\bar{P}_\infty = P_\infty + P_\infty (\gamma_f^2 I - P_\infty)^{-1} P_\infty$ in accordance to Assumption~\ref{ass:terminal_cost}.

Concerning the design of the terminal  sets $\mathcal{X}_{f,2}$ and $\mathcal{X}_{f,\infty}$, we endow each scheme with a robust positively invariant polyhedral outer-approximation of the minimal robust positively invariant set under the corresponding local controller $K_2$ or $K_\infty$. To do so, we follow Algorithm~1 in \cite{rakovic2005invariant}, with approximation parameter $\epsilon = 0.5$. Importantly, we verify numerically that $\mathcal{X}_{f,2}$ and $\mathcal{X}_{f,\infty}$ constitute valid terminal sets, since each of them is constraint admissible under the associated local controller as required by Assumption~\ref{ass:terminal_set}.

We remark that, by construction, the receding horizon control laws with $\mathcal{H}_2$ and $\mathcal{H}_\infty$ objectives are recursively feasible and render the closed-loop system input-to-state and $\ell_2$ stable, respectively \cite[Th. 4.20, Th. 5.15]{goulart2007affine}.

\emph{2) Schemes based on regret minimization:} In line with Assumptions~\ref{ass:terminal_cost}-\ref{ass:terminal_set}, we employ the terminal ingredients $P_\infty$ and $\mathcal{X}_{f, \infty}$ derived from the auxiliary $\mathcal{H}_\infty$ control policy previously computed. Moreover, to evaluate the term $\mathbf{C}_T$ that appears in \eqref{prob:safe_regret_optimal_sdp_bounded_energy}, we first compute the closed-loop responses associated with the finite horizon clairvoyant optimal policy $\bm{\pi}^c_T$ by solving the quadratic optimization problem \eqref{eq:clairvoyant_policy_sls}. In \eqref{eq:sdp_schur_constraints_bounded_energy}, we use $\sigma = 1$.

\emph{3) Performance comparison:}
We now compare the average control cost incurred by the receding horizon $\mathcal{H}_2$, $\mathcal{H}_\infty$, and regret-optimal ($\mathcal{R}$) receding horizon safe control policies. We draw the disturbance realizations according to a variety of stochastic and deterministic profiles, including a scaled Gaussian distribution and moving average sinusoidal signals. For each of these scenarios, we simulate the closed-loop system behavior when the number $s$ of control actions applied before the optimization is repeated varies from $1$ to $T$. We collect our results in Table~\ref{table:results};\footnote{The code that reproduces our numerical examples is available at \href{https://github.com/DecodEPFL/RHSafeMinRegret}{https://github.com/DecodEPFL/RHSafeMinRegret}. Please refer to our source code for a precise definition of the disturbance profiles that appear in Table~\ref{table:results}.} note that we only report a single column for the receding horizon scheme with $\mathcal{H}_2$ objective, as we observe that the solution rapidly converges to the unconstrained optimal controller $K_2$.

We verify that, as predicted by \cite[Th. 4.20, Prop. 5.14]{goulart2007affine} for receding horizon $\mathcal{H}_2$ and $\mathcal{H}_\infty$ control, respectively, and by Theorem~\ref{th:recursive_feasibility_terminal_rpi} for the case of regret minimization, all predictive control schemes are recursively feasible. Also, we observe numerically that $\mathcal{X}_{f,\infty} \subset \mathcal{X}_{f,2}$, meaning that the receding horizon $\mathcal{H}_2$ policy yields the largest feasible set and domain of attraction.

Regarding performance, we observe that using an $\mathcal{H}_2$ objective leads to the lowest control costs when the disturbance follows a Gaussian distribution. This is expected since this is precisely the class of disturbances that $\mathcal{H}_2$ controllers are designed to counteract. Instead, despite having imposed more challenging terminal constraints, minimizing regret can yield superior performance for other practically relevant disturbance profiles, such as piece-wise constant or sinusoidal signals, that poorly fit classical design assumptions.

We found that an advantage of regret-optimal policies facing non-stochastic disturbances is that a consistenly lower control cost is attained despite optimizing less frequently, i.e., when $s \in \mathbb{I}_{[1, T]}$ is large. This result may appear counter-intuitive at first; the explanation lies in the inherent capability of regret-optimal policies to adapt to heterogeneous disturbance sequences. Indeed, unlike $\mathcal{H}_2/\mathcal{H}_\infty$-optimal policies, regret-optimal policies exploit the full history of signals even in the unconstrained case in order to favour adaptivity. In the spirit of event-triggered predictive control \cite{lehmann2013event}, we hope that regret-optimal predictive control will especially thrive in applications where communication rates or energy consumption are limited, requiring a significant reduction of the frequency at which the planning optimization problem has to be solved.

On the other extreme, when $s = 1$, we numerically observe that the performance matches that of a receding horizon architecture with $\mathcal{H}_2$ objective. This consideration is consistent with the result of \cite[Th. 4]{sabag2021regret}, which shows that unconstrained regret-optimal policies have a state-feedback law of the $\mathcal{H}_2$ controller plus a correction term obtained by linearly combining the observed disturbance realizations. By optimizing at each time-step without storing additional information about the entire sequence of observed disturbance realizations other than that provided by the state vector at the current time instant, the above-mentioned correction term becomes zero. On the other hand, introducing memory can prove challenging as the state augmentation procedure proposed in \cite{goel2023regret} to solve a finite horizon regret-optimal control problem yields time-varying dynamics governed by a series of Riccati recursions, and we leave this extension as an interesting direction for future research.

\emph{4) Robustness of performance comparison:}
To confirm the robustness of our conclusions, we perform a second set of experiments by focusing on a family of sinusoidal disturbance signals $w_t$. In particular, for $i \in \mathbb{I}_3$, we construct each individual disturbance component $w_t^i = \operatorname{sin}(\omega t + \varphi)$ by choosing 10 different values for the angular frequency $\omega$ and the phase $\varphi$ equally spaced in the intervals $[\frac{3\pi}{N}, \frac{12\pi}{N}] = [0.1571, 0.6283]$ and $[0, 2\pi]$, respectively. We plot the results obtained by simulating the closed-loop system behavior in Figure~\ref{fig:sin_varying}. We visualize the standard deviation of the energy-normalized average control cost $\bar{J}$ across different angular frequencies or phases by means of shaded tubes centered around the mean values. We observe that changing the angular frequency impacts the realized closed-loop performance more, yet our receding horizon control scheme based on regret minimization can outperform classical architectures consistently across nearly all sinusoidal disturbances. Figure~\ref{fig:sin_varying} further validates that the adaptation capabilities inherent to regret-optimal policies can prove key when the disturbances are non-stochastic. In fact, as $s$ decreases, we note that our scheme again approaches the performance of the receding horizon $\mathcal{H}_2$ control law, while the best performance are attained by adopting an $\mathcal{H}_\infty$ objective.

\begin{figure*}[htb]
    \begin{subfigure}{\columnwidth}
      \centering
    \includegraphics[width=.95\columnwidth]{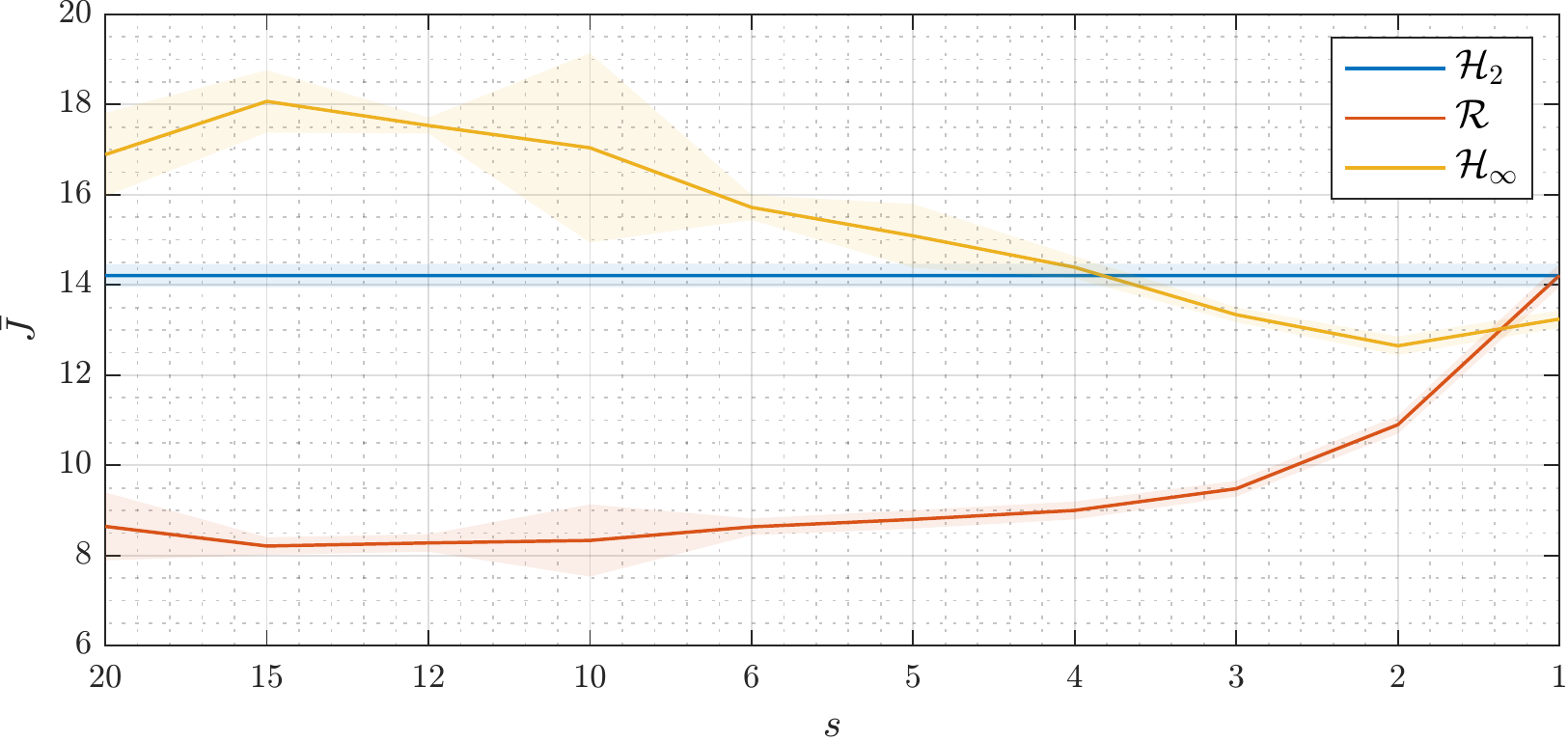}
        \caption{\centering Effect of varying the phase $\varphi$ in $[0, 2\pi]$.}
        \label{fig:sin_varying_phase}
    \end{subfigure}
    \begin{subfigure}{\columnwidth}
        \centering
        \includegraphics[width=.95\columnwidth]{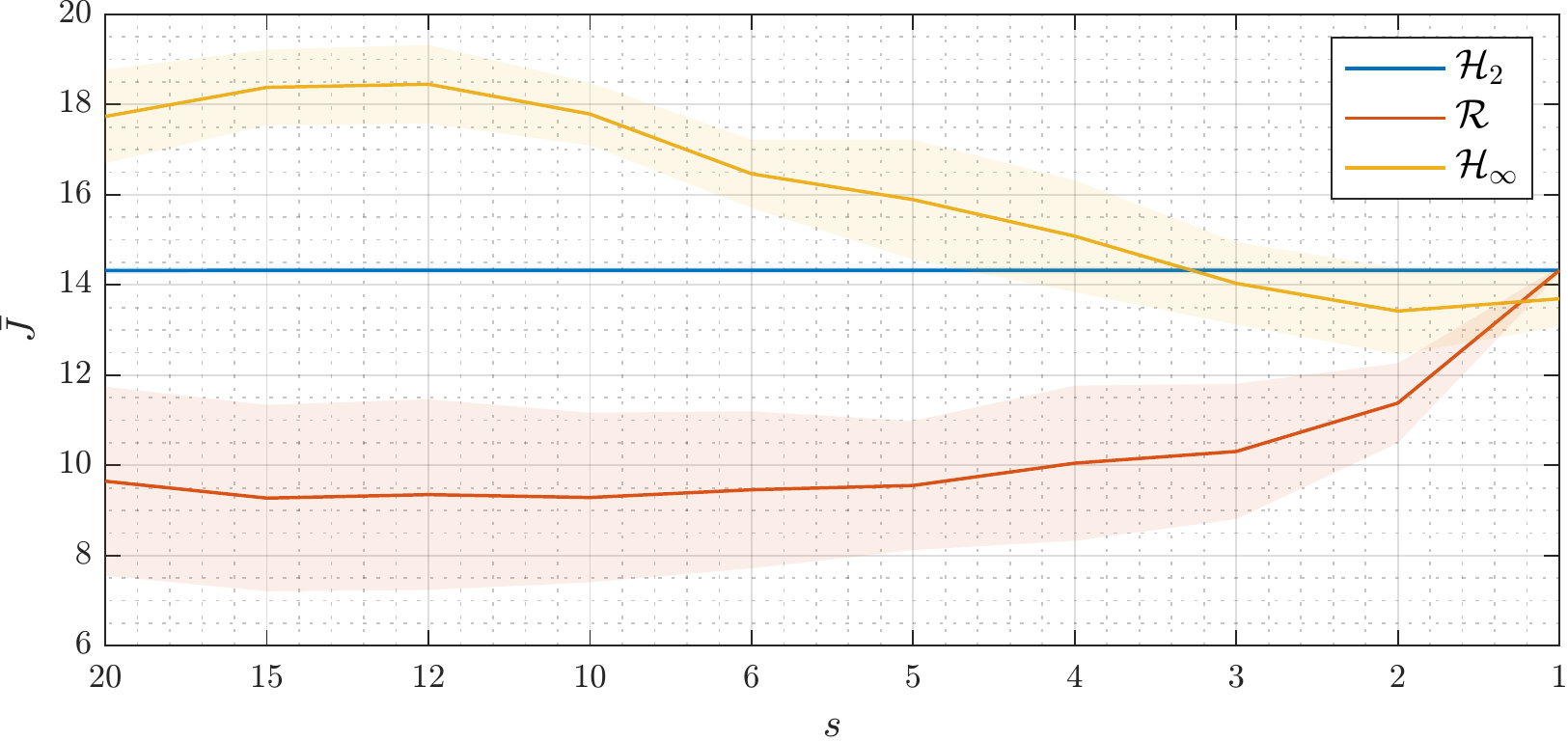}
        \caption{\centering Effect of varying the angular frequency $\omega$ in $[\frac{3\pi}{N}, \frac{12\pi}{N}]$.}
        \label{fig:sin_varying_freq}
    \end{subfigure}
    \caption{Normalized control cost incurred by the $\mathcal{H}_2$, $\mathcal{H}_\infty$, and regret-optimal ($\mathcal{R}$) receding horizon safe control policy in correspondence with sinusoidal disturbance profiles with different angular frequencies and phases: shaded areas denote the standard deviations around mean values.}
    \label{fig:sin_varying}
\end{figure*}

\section{Conclusion}
\label{sec:conclusion}
We studied control of constrained linear systems with bounded additive disturbances from the perspective of receding horizon regret minimization. By imposing terminal constraints derived from an auxiliary unconstrained $\mathcal{H}_\infty$ control law, we have formally shown that repeatedly minimizing dynamic regret guarantees recursive feasibility of the predictive control scheme and $\ell_2$ stability of the resulting closed-loop system. In the process of establishing these theoretical certificates, we have derived novel monotonicity properties of the clairvoyant optimal policy; an interesting direction for future research is to combine these with a precise analysis of the regret-to-go at each time-step to derive tighter infinite horizon regret guarantees. %
Additionally, future work includes developing model-free solutions, addressing computational complexity challenges for real-time implementation, deriving analogous results for estimation, and studying the interplay between different competitive metrics for nonlinear systems.

bibliographystyle{IEEEtran}
\bibliography{references}

\appendix
\allowdisplaybreaks

\subsection{Proof of Lemma~\ref{le:clairvoyant_optimal_policy}}
Recall that $\mathbf{F} = (\mathbf{I} - \mathbf{Z} \mathbf{A})^{-1} \mathbf{Z} \mathbf{B}$ and $\mathbf{G} = (\mathbf{I} - \mathbf{Z} \mathbf{A})^{-1}$ denote the causal response operators that encode the system dynamics as $\mathbf{x} = \mathbf{F}\mathbf{u} + \mathbf{G} \bm{\delta}$. Then, the control cost an arbitrary decision policy $\bm{\pi}_T$ incurs over $T$ steps can be expressed as
\begin{align}
    J_T(\bm{\pi}_T, \bm{\delta})
    &= \left(\mathbf{F} \mathbf{u} + \mathbf{G} \bm{\delta} \right)^\top \mathbf{Q} \left(\mathbf{F} \mathbf{u} + \mathbf{G} \bm{\delta} \right) + \mathbf{u}^\top \mathbf{R} \mathbf{u} \nonumber\\
    &= \mathbf{u}^\top \mathbf{P} \mathbf{u} + 2 \mathbf{u}^\top \mathbf{F}^\top \mathbf{Q} \mathbf{G} \bm{\delta} + \bm{\delta}^\top \mathbf{G}^\top \mathbf{Q} \mathbf{G} \bm{\delta} \label{eq:lqr_cost_compact_FG}\,,
\end{align}
where $\mathbf{P} = \mathbf{R} + \mathbf{F}^\top \mathbf{Q}\mathbf{F} \succ 0$. Further, completing the square by observing that
\begin{equation*}
    \mathbf{Q}\mathbf{F}\mathbf{P}^{-1}\mathbf{F}^\top \mathbf{Q} + \mathbf{Q}(\mathbf{I} + \mathbf{F}\mathbf{R}^{-1}\mathbf{F}^\top\mathbf{Q})^{-1} = \mathbf{Q}\,,
\end{equation*}
thanks to the Woodbury matrix identity, the cumulative control cost in \eqref{eq:lqr_cost_compact_FG} can be equivalently rewritten as
\begin{align}
    J_T(\bm{\pi}_T, \bm{\delta}) &= [\mathbf{P}\mathbf{u} + \mathbf{F}^\top \mathbf{Q} \mathbf{G} \bm{\delta}]^\top
    \mathbf{P}^{-1}
    [\mathbf{P}\mathbf{u} + \mathbf{F}^\top \mathbf{Q} \mathbf{G} \bm{\delta}] + \nonumber\\
    &+ \bm{\delta}^\top \mathbf{G}^\top \mathbf{Q} (\mathbf{I} + \mathbf{F}\mathbf{R}^{-1}\mathbf{F}^\top \mathbf{Q})^{-1} \mathbf{G} \bm{\delta}\label{eq:lqr_cost_compact_split}\,.
\end{align}
We then note that the first addend in \eqref{eq:lqr_cost_compact_split} is non-negative, as $\mathbf{P} \succ 0$ by construction, and that the second term does not depend on the control actions that the decision policy $\bm{\pi}_T$ selects. Hence, in the absence of safety constraints, the cost-minimizing control sequence is obtained by setting the former term is equal to zero. The expression \eqref{eq:clairvoyant_controller_definition} for $\bm{\pi}^c_T$ is then obtained by direct computation. Further, \eqref{eq:clairvoyant_controller_cost_incurred} is simply the second addend in \eqref{eq:lqr_cost_compact_split}, which directly yields the control cost of the policy $\bm{\pi}^c_T$ incurs when the disturbance is $\bm{\delta}$.

\subsection{Proof of Proposition~\ref{prop:convexity_lower_bound}}
\label{app:proof_convexity_lower_bound}
    For any fixed performance level $\gamma \in \mathbb{R}_{\geq 0}$ and perturbation sequence $\mathbf{w} \in \mathbfcal{W}$, the function $(\bm{\Phi}_x, \bm{\Phi}_u) \mapsto \mathtt{Reg}_T(\bm{\Phi}_x, \bm{\Phi}_u, \gamma, \bm{\delta})$ is convex and continuous by inspection, see also \cite{martin2022safe}. Convexity and lower semicontinuity of the function $(\bm{\Phi}_x, \bm{\Phi}_u) \mapsto \mathtt{Reg}_T^\star(\bm{\Phi}_x, \bm{\Phi}_u, \gamma, x_0)$ then follow, since the point-wise maximum of lower semicontinuous and convex functions is lower semicontinuous and convex \cite[Prop. 2.9]{rockafellar2009variational}.
    To prove that this function is bounded from below, we recall that $V_f(x_T) \geq 0$ for any $x_T$ since $P \succeq 0$ by definition, and we note that the optimality of the clairvoyant policy $\bm{\pi}^c_T$ implies that $J_T(\bm{\pi}_T, \bm{\delta}) - J_T(\bm{\pi}^c_T, \bm{\delta}) \geq 0$ for any policy $\bm{\pi}_T$ and disturbance $\mathbf{w}$. Hence, we have that
    \begin{equation*}
         \norm{\bar{\mathbf{S}}^{\frac{1}{2}} \begin{bmatrix}
            \bm{\Phi}_x \\ \bm{\Phi}_u
        \end{bmatrix} \bm{\delta}}_2^2
        -
        \norm{\mathbf{S}^{\frac{1}{2}} \begin{bmatrix}
            \bm{\Phi}_x^c \\ \bm{\Phi}_u^c
        \end{bmatrix} \bm{\delta}}_2^2 \geq 0\,.
    \end{equation*}
    Further, since $\mathbfcal{W}$ contains an open neighborhood of the origin by assumption, we have that
    \begin{align*}
        \mathtt{Reg}_T^\star(\bm{\Phi}_x, \bm{\Phi}_u, \gamma, x_0) &\geq \mathtt{Reg}_T(\bm{\Phi}_x, \bm{\Phi}_u, \gamma, x_0, \mathbf{0}) \geq 0\,,
    \end{align*}
    which concludes the proof of the lower bound. Finally, \eqref{eq:objective_maximization} is also proper, i.e., it only takes finite values for every $\{\bm{\Phi}_x, \bm{\Phi}_u\}$, since the set $\mathbfcal{W}$ is assumed compact.

\subsection{Proof of Lemma~\ref{le:relations_clairvoyant}}
\label{app:proof_relations_clairvoyant}
    We first consider the left inequality $\mathbf{C}_{T+1}^{[1:T+1, 1:T+1]} \preceq \mathbf{C}_{T}$. We equivalently rewrite the cost that the clairvoyant policy $\bm{\pi}_{T+1}^c$ incurs over a planning horizon of length $T+1$ as
    \begin{equation*}
        x_0^\top C_{T+1}^{[0, 0]} x_0 + 2 x_0^\top \mathbf{C}_{T+1}^{[0, 1:T+1]} \bar{\mathbf{w}} + \bar{\mathbf{w}}^\top \mathbf{C}_{T+1}^{[1:T+1, 1:T+1]} \bar{\mathbf{w}}\,,
    \end{equation*}
    where $\bar{\mathbf{w}} = \operatorname{col}(w_0, w_1, \dots, w_T)$ collects the realizations of the disturbance process along the extended planning horizon. By inspection, it follows that the addend $\bar{\mathbf{w}}^\top \mathbf{C}_{T+1}^{[1:T+1, 1:T+1]} \bar{\mathbf{w}}$ constitutes the entire cumulative cost incurred by the clairvoyant policy $\bm{\pi}^c_{T+1}$
    if the system is initialized at $x_0 = 0$, i.e.,
    \begin{equation*}
        \bar{\mathbf{w}}^\top \mathbf{C}_{T+1}^{[1:T+1, 1:T+1]} \bar{\mathbf{w}} = J_{T+1}(\bm{\pi}_{T+1}^c, 0, \mathbf{w}_{0:T})\,.
    \end{equation*}
    On the other hand, for any $\bar{\mathbf{w}}$, we also have that $\bar{\mathbf{w}}^\top \mathbf{C}_T \bar{\mathbf{w}} = J_{T}(\bm{\pi}_{T}^c, w_{0}, \mathbf{w}_{1:T})$,
    i.e., the cost incurred by the clairvoyant policy $\bm{\pi}_{T}^c$ over a planning horizon of length $T$ when the system is initialized at $x_{1} = w_{0}$ and the disturbance sequence $\mathbf{w}_{1:T}$ realizes. Then, to conclude the first part of the proof, we show that
    \begin{equation}
        \label{eq:final_subsequences_cost_inequality}
        J_{T}(\bm{\pi}_{T}^c, w_{0}, \mathbf{w}_{1:T}) \geq J_{T+1}(\bm{\pi}_{T+1}^c, 0, \mathbf{w}_{0:T})\,,
    \end{equation}
    for any $\bar{\mathbf{w}}$. To do so, let us denote by $\mathbf{u}^\star_{1:T}$ the sequence of clairvoyant control actions selected in hindsight by the policy $\bm{\pi}_{T}^c$ to minimize the left-hand side of \eqref{eq:final_subsequences_cost_inequality}. Also, let $\mathbf{x}^\star_{1:T+1}$ be the corresponding state trajectory, initialized at $x^\star_{1} = w_{0}$. Then, the clairvoyant control sequence defined by
    \begin{equation}
        \label{eq:final_subsequences_input_feasible}
        \bar{u}_0 = 0\,, ~~ \bar{u}_t = u_t^\star\,, ~ \forall t \in \mathbb{I}_{[1, T]}\,,
    \end{equation}
    can be implemented by clairvoyant policy $\bm{\pi}_{T+1}^c$, which also has complete foreknowledge of the disturbance realizations $\mathbf{w}_{0:T}$, and results, when the system is initialized at $x_0 = 0$ under $\bar{\mathbf{w}}$, in the state trajectory
    \begin{equation}
        \label{eq:final_subsequences_state_feasible}
        \bar{x}_0 = 0\,, ~~ \bar{x}_t = x_t^\star\,, ~ \forall t \in \mathbb{I}_{[1, T+1]}\,.
    \end{equation}
    It is then straightforward to see that
    \eqref{eq:final_subsequences_input_feasible} and \eqref{eq:final_subsequences_state_feasible} imply that
    \begin{equation*}
        \sum_{k = 1}^{T} {x^\star_k}^\top Q x^\star_k + {u^\star_k}^\top R u_k^\star =
        \sum_{k = 0}^{T} \bar{x}_k^\top Q \bar{x}_k + \bar{u}_k^\top R \bar{u}_k\,,
    \end{equation*}
    namely that, when $x_0 = 0$, the feasible control sequence \eqref{eq:final_subsequences_input_feasible} attains over $T+1$ a cost equal to $J_{T}(\bm{\pi}_{T}^c, w_{0}, \mathbf{w}_{1:T})$. The left inequality in \eqref{eq:clairvoyant_inequalities_chain} then follows by suboptimality of $\bar{\mathbf{u}}_{0:T}$ with respect to $\bm{\pi}_{T+1}^c(\operatorname{col}(0, \bar{\mathbf{w}}))$.

    We next consider the right inequality $\mathbf{C}_{T} \preceq \mathbf{C}_{T+1}^{[0:T, 0:T]}$. We start by rewriting the cost incurred over $T+1$ by the clairvoyant policy $\bm{\pi}_{T+1}^c$ as
    \begin{equation*}
        \ubar{\bm{\delta}}^\top \mathbf{C}_{T+1}^{[0:T, 0:T]} \ubar{\bm{\delta}}
        + 2 \ubar{\bm{\delta}}^\top \mathbf{C}_{T+1}^{[0:T, T+1]} w_{T} + w_{T}^\top C_{T+1}^{[T+1, T+1]} w_{T}\,,
    \end{equation*}
    where $\ubar{\bm{\delta}} = \operatorname{col}(x_0, w_0, \dots, w_{T-1})$ now collects the system initial condition and all disturbance realizations but the last one. By inspection, whenever $w_T = 0$, the previous expression simplifies to $\ubar{\bm{\delta}}^\top \mathbf{C}_{T+1}^{[0:T, 0:T]} \ubar{\bm{\delta}}$, namely
    \begin{equation*}
        \ubar{\bm{\delta}}^\top \mathbf{C}_{T+1}^{[0:T, 0:T]} \ubar{\bm{\delta}} = J_{T+1}(\bm{\pi}_{T+1}^c, x_0, \mathbf{w}_{0:T-1}, 0)\,.
    \end{equation*}
    On the other hand, for any $\ubar{\bm{\delta}}$, we also have that $\ubar{\bm{\delta}}^\top \mathbf{C}_T \ubar{\bm{\delta}} = J_{T}(\bm{\pi}_{T}^c, x_{0}, \mathbf{w}_{0:T-1})$, i.e., the cost incurred by the clairvoyant policy $\bm{\pi}_{T}^c$ over a planning horizon of length $T$ when the system is initialized at $x_{0} \in \mathbb{R}^n$ and the disturbance sequence $\mathbf{w}_{0:T-1}$ realizes. To conclude the second part of the proof, we show that
    \begin{subequations}
        \label{eq:initial_subsequences_cost_inequality}
        \begin{align}
            J_{T}(\bm{\pi}_{T}^c, x_{0}, \mathbf{w}_{0:T-1}) &\leq J_{T}(\bm{\pi}_{T+1}^c, x_0, \mathbf{w}_{0:T})\label{eq:initial_subsequences_cost_inequality_1}\\
            &\leq J_{T+1}(\bm{\pi}_{T+1}^c, x_0, \mathbf{w}_{0:T})\label{eq:initial_subsequences_cost_inequality_2}\,,
        \end{align}
    \end{subequations}
    for any possible choices of $\ubar{\bm{\delta}}$ and $w_T$, and in particular for the case of interest $w_T = 0$. We start by noting that, since the dynamics of the system under control are causal, the cost that any policy incurs during the planning horizon $T$ does not explicitly depend on $w_T$ – see also \eqref{eq:lqr_cost_compact_split}, which depends on $\mathbf{w}_{0:T-1}$ but not on $w_T$. We then deduce that the globally optimal sequence of control actions that minimizes \eqref{eq:lqr_cost_compact_split} also does not depend on $w_T$. Hence, \eqref{eq:initial_subsequences_cost_inequality_1} follows by optimality of $\bm{\pi}^c_T$, since knowing the realization of $w_T$ in hindsight does not influence the cost-minimizing sequence of control inputs. Finally, the inequality \eqref{eq:initial_subsequences_cost_inequality_2} follows, since the linear quadratic stage cost $x_T^\top Q x_T + u_T^\top R u_T$ is non-negative by definition.

\subsection{Proof of Corollary~\ref{cor:eigenvalues_clairvoyant}}
\label{app:proof_eigenvalues_clairvoyant}
    Let the plant under control evolve under the clairvoyant optimal policy $\bm{\pi}^c_T$, and let $\ubar{\bm{\delta}} = \operatorname{col}(x_0, w_0, \dots, w_{T-1})$ denote the disturbance sequence with unit norm that results in the highest control cost, i.e., $\ubar{\bm{\delta}} = \operatorname{argmax}_{\norm{\bm{\delta}}_2 = 1} \bm{\delta}^\top \mathbf{C}_T \bm{\delta}$ and $\ubar{\bm{\delta}}^\top \mathbf{C}_T \ubar{\bm{\delta}} = \lambda_{\operatorname{max}}(\mathbf{C}_T)$. Consider now the unitary norm extended perturbation sequence $\tilde{\bm{\delta}}$ defined by $\tilde{\bm{\delta}} = \operatorname{col}(\ubar{\bm{\delta}}, 0)$. Since $\mathbf{C}_{T} \preceq \mathbf{C}_{T+1}^{[0:T, 0:T]}$ in light of Lemma~\ref{le:relations_clairvoyant}, we have that $\lambda_{\operatorname{max}}(\mathbf{C}_T) = \ubar{\bm{\delta}}^\top \mathbf{C}_T \ubar{\bm{\delta}} \leq \tilde{\bm{\delta}}^\top \mathbf{C}_{T+1} \tilde{\bm{\delta}}$. By suboptimality of $\tilde{\bm{\delta}}$ with respect to $\operatorname{argmax}_{\norm{\bm{\delta}}_2 = 1} \bm{\delta}^\top \mathbf{C}_{T+1} \bm{\delta}$, we also have that $\tilde{\bm{\delta}}^\top \mathbf{C}_{T+1} \tilde{\bm{\delta}} \leq \lambda_{\operatorname{max}}(\mathbf{C}_{T+1})$,
    and the claimed monotonicity property follows. The second part of the proof regarding the minimum eigenvalue follows similar derivations, and is thus omitted.

\subsection{Proof of Proposition~\ref{prop:chain_of_inclusions}}
\label{app:proof_chain_of_inclusions}
    We prove this chain of set inclusions by induction, and consider the base case $\mathcal{Z}_f \subseteq \mathcal{X}_1^\gamma$ first. For any $x_0 \in \mathcal{Z}_f$, the local control law $u_0 = K_f x_0$ can be used to safely maintain the system inside the terminal set while satisfying a quadratic matrix inequality $\mathbf{H}_1^\gamma(K_f) \succeq 0$ of the form \eqref{eq:h_infinity_objective_qmi}, as shown in \cite[Prop. 3]{goulart2009control}. Hence, since $\mathbf{C}_{1}^{[1,1]} \succeq 0$, we conclude that the same candidate policy satisfies a quadratic matrix inequality $\mathbf{R}_1^\gamma(K_f) \succeq 0$ of the form \eqref{eq:regret_objective_qmi}.

    For the induction step, we show that an arbitrary initial condition $x_0 \in \mathcal{X}_T^\gamma$ satisfies $x_0 \in \mathcal{X}_{T+1}^\gamma$. By definition, for any $x_0 \in \mathcal{X}_T^\gamma$, there exists an admissible control policy at performance level $\gamma$ $\bm{\pi}_T \in \bm{\Pi}_{T}^\gamma(x_0, \gamma)$ that drives the system to a terminal state $x_T \in \mathcal{Z}_f$ while satisfying $(x_t, u_t) \in \mathcal{Z}$ along the way. If Assumptions~\ref{ass:stab_unit_circle_obs}-\ref{ass:terminal_set} hold with $\gamma_f \leq \gamma$, then an admissible control $\bm{\pi}_{T+1} \in \bm{\Pi}_{T+1}(x_0)$ can be constructed leveraging the admissible and robust positively invariant local controller $K_f$, namely, appending the input $u_T = \pi_T(x_0, \dots, x_T) = K_f x_T$ to the control sequence in $\bm{\pi}_T$. It only remains to show that this proposed candidate policy $\bm{\pi}_{T+1}$ also satisfies a quadratic matrix inequality of the form \eqref{eq:regret_objective_qmi} with the same performance level $\gamma$, i.e., that $\bm{\pi}_{T+1} \in \bm{\Pi}_{T+1}^\gamma(x_0)$. To do so, in light of the equivalence between SLS and disturbance feedback parametrizations \cite{zheng2020equivalence}, we recall from \cite[Prop. 3]{goulart2009control} that $\mathbf{H}_{T+1}^\gamma(\bm{\pi}_{T+1})$, the $\mathcal{H}_\infty$-type quadratic form associated with the proposed candidate policy $\bm{\pi}_{T+1}$, can be related to $\mathbf{H}_T^\gamma(\bm{\pi}_T)$ through
    \begin{equation*}
        \mathbf{H}_{T+1}^\gamma(\bm{\pi}_{T+1}) = \begin{bmatrix}
            \mathbf{H}_T^\gamma(\bm{\pi}_T) + \mathbf{X} & -\mathbf{Y}^\top A_f^\top P\\
            - P A_f \mathbf{Y} & \gamma^2 I - P
        \end{bmatrix}\,,
    \end{equation*}
    where $A_f = A + BK_f$, for appropriately defined matrices $\mathbf{X}$ and $\mathbf{Y}$. Observing that, by definition, $\mathbf{R}_{T+1}^\gamma(\bm{\pi}_{T+1}) = \mathbf{H}_{T+1}^\gamma(\bm{\pi}_{T+1}) + \mathbf{C}_{T+1}^{[1:T+1, 1:T+1]}$, we rewrite $\mathbf{R}_{T+1}^\gamma(\bm{\pi}_{T+1})$ as
    \begin{equation*}
        \begin{bmatrix}
            \mathbf{H}_T^\gamma(\pi_T) + \mathbf{X} & -\mathbf{Y}^\top A_f^\top P\\
            -P A_f \mathbf{Y} & \gamma^2 I - P
        \end{bmatrix} + \mathbf{C}_{T+1}^{[1:T+1, 1:T+1]}\,.
    \end{equation*}
    Therefore, we deduce that $\mathbf{R}_{T+1}^\gamma(\bm{\pi}_{T+1})$ and $\mathbf{R}_T^\gamma(\bm{\pi}_T)$ can be related through
    \begin{align}
    \label{eq:set_inclusion_regret_quadratic_form_relation_initial}
        \mathbf{R}_{T+1}^\gamma(\bm{\pi}_{T+1}) = \begin{bmatrix}
            \mathbf{R}_T^\gamma(\pi_T) + \mathbf{X} & -\mathbf{Y}^\top A_f^\top P\\
            -P A_f \mathbf{Y} & \gamma^2 I - P
        \end{bmatrix}
        + \bm{\Delta}_{i}\,,
    \end{align}
    where $\bm{\Delta}_i = \mathbf{C}_{T+1}^{[1:T+1, 1:T+1]} - \operatorname{blkdiag}(\mathbf{C}_T^{[1:T, 1:T]}, 0_{n \times n}) $ captures the clairvoyant baseline shift that occurs as planning horizons of different length are considered. Proceeding analogously to \cite[Prop. 3]{goulart2009control}, one can show that $\mathbf{R}_T^\gamma(\bm{\pi}_T) \succeq 0$ and $\gamma^2 I - P \succ 0$, which hold true by assumption, imply that the first addend in \eqref{eq:set_inclusion_regret_quadratic_form_relation_initial} is positive semidefinite. Hence, it suffices to show that $\bm{\Delta}_i \succeq 0$ to conclude that $\mathbf{R}_{T+1}^\gamma(\bm{\pi}_{T+1}) \succeq 0$. To see this, we observe that for any perturbation sequence $\mathbf{w}_{0:T}$, the quadratic form $\bar{\mathbf{w}}^\top \bm{\Delta}_i \bar{\mathbf{w}}$, with $\bar{\mathbf{w}} = \operatorname{col}(w_0, \dots, w_T)$, can be equivalently computed as $J_{T+1}(\bm{\pi}_{T+1}^c, w_0, \mathbf{w}_{1:T}) - J_{T}(\bm{\pi}_{T}^c, w_{0}, \mathbf{w}_{1:T-1})$.
    The positive semidefiniteness of $\bm{\Delta}_i$ then follows from \eqref{eq:initial_subsequences_cost_inequality} in Lemma~\ref{le:relations_clairvoyant}, and the proof is completed.

\subsection{Proof of Theorem~\ref{th:recursive_feasibility_terminal_rpi}}
\label{app:proof_recursive_feasibility_terminal_rpi}
    If Assumptions~\ref{ass:stab_unit_circle_obs}-\ref{ass:terminal_set} hold with performance level $\gamma \geq \operatorname{max}\{\gamma_T^\star(x_0), \gamma_f\}$, then Proposition~\ref{prop:chain_of_inclusions} guarantees that, for any initial condition $x_0 \in \mathcal{X}_T^\gamma$, there exists an admissible control policy $\bm{\pi}_{T+1} \in \bm{\Pi}^\gamma_{T+1}(x_0) \subseteq \bm{\Pi}_{T+1}(x_0)$ with first control action $\mu_T(x_0, \gamma)$. Let $x_1 = A x_0 + B\mu_T(x_0, \gamma) + w_0$, and define $\bm{\pi}_T$ as the tail of length $T$ extracted from the control policy $\bm{\pi}_{T+1}$. By construction, $\bm{\pi}_T \in \bm{\Pi}_T(x_1)$; it thus only remains to show that $\bm{\pi}_T \in \bm{\Pi}_T^\gamma(x_1)$, i.e., the performance level $\gamma$ is met. To do so, we recall from \cite[Lem. 1]{goulart2009control} that $\mathbf{H}_T^\gamma(\bm{\pi}_T)$ in \eqref{eq:h_infinity_objective_qmi} can be related to $\mathbf{H}_{T+1}^\gamma(\bm{\pi}_{T+1})$ by
    \begin{equation}
        \label{eq:set_inclusion_h_infinity_quadratic_form_relation_final}
        \mathbf{H}_{T+1}^\gamma(\bm{\pi}_{T+1}) = \begin{bmatrix}
            U & \mathbf{V}\\
            \mathbf{V}^\top & \mathbf{H}_T^\gamma(\bm{\pi}_T)
        \end{bmatrix}\,,
    \end{equation}
    for appropriately defined matrices $U$ and $\mathbf{V}$. We then deduce that $\mathbf{R}_T^\gamma(\pi_T)$ in \eqref{eq:regret_objective_qmi} and $\mathbf{R}_{T+1}^\gamma(\pi_{T+1})$ can be related by
    \begin{equation}
        \label{eq:set_inclusion_regret_quadratic_form_relation_final}
        \mathbf{R}_{T+1}^\gamma(\bm{\pi}_{T+1}) = \begin{bmatrix}
            U & \mathbf{V}\\
            \mathbf{V}^\top & \mathbf{R}_T^\gamma(\bm{\pi}_{T})
        \end{bmatrix}
        +
        \bm{\Delta}_f \succeq 0\,,
    \end{equation}
    where $\bm{\Delta}_f = \mathbf{C}_{T+1}^{[1:T+1, 1:T+1]} - \operatorname{blkdiag}(0_{n \times n}, \mathbf{C}_{T}^{[1:T, 1:T]})$ captures the clairvoyant baseline shift that occurs as planning horizons of different length are considered. By partitioning $\bm{\Delta}_f$ into blocks of appropriate dimensions, we equivalently rewrite \eqref{eq:set_inclusion_regret_quadratic_form_relation_final} as
    \begin{equation*}
        \begin{bmatrix}
            U + C_{T+1}^{[1, 1]} & \bar{\mathbf{V}}\\
            \bar{\mathbf{V}}^\top & \mathbf{R}_T^\gamma(\bm{\pi}_{T})
        \end{bmatrix}
        +
        \begin{bmatrix}
            0 & 0\\
            0 & \mathbf{C}_{T+1}^{[2:T+1, 2:T+1]} - \mathbf{C}_{T}^{[1:T, 1:T]}
        \end{bmatrix}\,,%
    \end{equation*}
    where $\bar{\mathbf{V}} = \mathbf{V} + \mathbf{C}_{T+1}^{[1, 2:T+1]}$. Note that Lemma~\ref{le:relations_clairvoyant} implies that the second addend in the expression above is negative semidefinite, since $\mathbf{C}_{T+1}^{[2:T+1, 2:T+1]} - \mathbf{C}_{T}^{[1:T, 1:T]} \preceq 0$. Therefore, we deduce that
    \begin{equation*}
        \begin{bmatrix}
            U + C_{T+1}^{[1, 1]} & \bar{\mathbf{V}}\\
            \bar{\mathbf{V}}^\top & \mathbf{R}_T^\gamma(\bm{\pi}_{T})
        \end{bmatrix} \succeq
        0\,.
    \end{equation*}
    Finally, since all principal sub-matrices of a positive semidefinite matrix are positive semidefinite, we conclude that $\mathbf{R}_T^\gamma(\bm{\pi}_T) \succeq 0$, i.e., that $\bm{\pi}_T \in \bm{\Pi}_T^\gamma(x_1)$.

\subsection{Proof of Theorem~\ref{th:rhc_finite_gain_infinite_horizon}}
\label{app:proof_rhc_finite_gain_infinite_horizon}
    For any real-valued scalar map $\varphi : \mathbb{R}^n \to \mathbb{R}$, we define the one-step difference function $\Delta \varphi : \mathbb{R}^{2n+m} \to \mathbb{R}$ as
    \begin{equation*}
        \Delta \varphi (x, u, w) = \varphi (Ax + Bu + w) - \varphi(x)\,.
    \end{equation*}
    Further, inspired by classical literature on $\mathcal{H}_\infty$ control, we let%
    \begin{align*}
        \ell(x_t, u_t, w_t) &= \norm{x_t}_Q^2 + \norm{u_t}_R^2 - \gamma^2\norm{w_t}_2^2\,,\\
        \ell_f(x_t, u_t, w_t) &= \norm{x_t}_Q^2 + \norm{u_t}_R^2 - \gamma_f^2\norm{w_t}_2^2\,,
    \end{align*}
    with $\norm{x_t}_Q^2 = x_t^\top Q x_t$ and $\norm{u_t}_R^2 = u_t^\top R u_t$, denote a quadratic stage loss that is negatively weighted in the energy of the disturbance realization. With this notation in place, the stage loss that measures the regret relative to the infinite horizon clairvoyant optimal policy $\bm{\pi}_\infty^c$ becomes
    \begin{align}
        \label{eq:lr_definition}
        \ell_r(x_t, u_t, w_t) = \ell(x_t, u_t, w_t)
        - \norm{x_t^c}^2_Q - \norm{u_t^c}^2_R\,,
    \end{align}
    where $u_t^c$ and $x_t^c$ denote the clairvoyant optimal control action at time $t$ and the corresponding state value, respectively.

    If Assumptions~\ref{ass:stab_unit_circle_obs}-\ref{ass:terminal_set} hold, the auxiliary feedback policy $K_f$, which is a saddle point solution of an unconstrained $\mathcal{H}_\infty$ zero-sum game, satisfies \cite{bacsar2008h}:
    \begin{equation}
        \label{eq:zero_sum_game_property}
        \max_{w_t \in \mathbb{R}^n} ~ \left[\Delta V_f + \ell_f\right](x_t, K_f x_t, w_t) = 0\,.
    \end{equation}
    Note that, for all $(x_t, u_t, w_t)$, we have that $\ell_r(x_t, u_t, w_t) \leq \ell(x_t, u_t, w_t) \leq \ell_f(x_t, u_t, w_t)$, where the first inequality follows because the stage cost incurred by the clairvoyant optimal policy is non-negative, and the second because $\gamma \geq \operatorname{max}\{\gamma_T^\star(x_0), \gamma_f\} \geq \gamma_f$ by assumption. Hence,
    \begin{equation*}
        \left[\Delta V_f + \ell_r\right](x_t, u_t, w_t) \leq \left[\Delta V_f + \ell_f\right](x_t, u_t, w_t)\,,
    \end{equation*}
    for all $(x_t, u_t, w_t)$. In particular, choosing $u_t = K_f x_t$ as in \eqref{eq:zero_sum_game_property} while restricting the domain of the maximization over the compact set $\mathcal{W} \subset \mathbb{R}^n$ only, we have that
    \begin{equation*}
        \label{eq:delta_terminal_cost_loss_regret_negative}
        \max_{w_t \in \mathcal{W}} ~ \left[\Delta V_f + \ell_r\right](x_t, K_f x_t, w_t) \leq 0\,,
    \end{equation*}
    for any $x_t$ and for any $w_t$, which in turn implies that $V_f$ is a robust control Lyapunov function in a neighborhood of the origin. Following standard arguments in robust receding horizon control \cite[Sect. 4]{mayne2000constrained}, this fact, together with the assumed robust invariance of $\mathcal{Z}_f$ under the local control law $u_t = K_f x_t$, ensures that
    \begin{equation}
        \label{eq:delta_value_function_loss_regret_negative}
        \left[\Delta V_T^\star(\cdot, \gamma) + \ell_r\right](x_t, \mu_T(x_t, \gamma), w_t) \leq 0\,,
    \end{equation}
    for all $x_t \in \mathcal{X}_T^\gamma$ and $w_t \in \mathcal{W}$. Taking the sum of the left-hand side of \eqref{eq:delta_value_function_loss_regret_negative} from time $0$ to a generic time $q \in \mathbb{N}$
    \begin{equation*}
        V_T^\star(x_q, \gamma) \leq V_T^\star(x_0, \gamma) - \sum_{k=0}^{q-1} \ell_r (x_k, \mu_T(x_k, \gamma), w_{k})\,,
    \end{equation*}
    Recalling the definition of $\ell_r(\cdot)$ in \eqref{eq:lr_definition} and since $V_T^\star(x_q, \gamma)$ is non-negative from Proposition~\ref{prop:convexity_lower_bound}, it then follows that
    \begin{align*}
        \sum_{t=0}^{q-1} \norm{x_t}_Q^2 + \norm{u_t}_R^2 &- \norm{x_t^c}^2_Q - \norm{u_t^c}^2_R \\
        &\leq \gamma^2 \sum_{t=0}^{q-1} \norm{w_t}_2^2 + V_T^\star(x_0, \gamma)\,,
    \end{align*}
    for any $q \in \mathbb{N}$, i.e., the amplification from the disturbance energy to the dynamic regret is finite and, in particular, is bounded above by $\gamma$. Lastly, Theorem~\ref{th:recursive_feasibility_terminal_rpi} ensures that the proposed receding horizon control law is recursively feasible, hence the constraints $(x_t, u_t) \in \mathcal{Z}$ are satisfied at all times.

\subsection{Proof of Proposition~\ref{prop:implementation_structured}}
    Let $\mathbf{M} = \mathbf{C}_{T}^{[0, 0]} - \bm{\Phi}_0^\top \bar{\mathbf{S}} \bm{\Phi}_0$, $\mathbf{N} = \mathbf{C}_{T}^{[0, 1:T]} - \bm{\Phi}_0^\top \bar{\mathbf{S}} \bm{\Phi}_w$, and $\mathbf{L} = \gamma^2 \mathbf{I} + \mathbf{C}_{T}^{[1:T, 1:T]} - \bm{\Phi}_w^\top \bar{\mathbf{S}} \bm{\Phi}_w$ so that
    \begin{align*}
        \mathtt{Reg}_T(\cdot) &= -\left(x_0^\top \mathbf{M} x_0  + 2 x_0^\top \mathbf{N} \mathbf{w} + \mathbf{w}^\top \mathbf{L}
        \mathbf{w}\right)\,, \\
        \mathtt{Reg}_T^\star(\cdot) &=
        -\min_{\mathbf{w} \in \mathbfcal{W}} ~ x_0^\top \mathbf{M} x_0  + 2 x_0^\top \mathbf{N} \mathbf{w} + \mathbf{w}^\top \mathbf{L}
        \mathbf{w}\,.
    \end{align*}
    We first observe that strong duality holds since $\mathbfcal{W}$ contains the origin in its interior by assumption. Therefore, we can equivalently rewrite $\mathtt{Reg}_T^\star(\bm{\Phi}_x, \bm{\Phi}_u, \gamma, x_0)$ as the optimal value of the following optimization problem
    \begin{align*}
        - \max_{\bm{\lambda} \geq 0} ~ \min_{\mathbf{w}} ~  x_0^\top \mathbf{M} x_0 + 2x_0^\top \mathbf{N} \mathbf{w} &+ \mathbf{w}^\top \mathbf{L} \mathbf{w}\\
        &+ \bm{\lambda}^\top (\mathbf{H}_w \mathbf{w} - \mathbf{h}_w)\,,
    \end{align*}
    or, equivalently, as the optimal value of
    \begin{align*}
        - \max_{\bm{\lambda} \geq 0} ~ x_0^\top \mathbf{M} x_0 &- \mathbf{h}_w^\top \bm{\lambda}\\ &+ \min_{\mathbf{w}} ~  (2\mathbf{N}^\top x_0 + \mathbf{H}_w^\top \bm{\lambda})^\top \mathbf{w} + \mathbf{w}^\top \mathbf{L} \mathbf{w}\,,
    \end{align*}
    to highlight the presence of an unconstrained minimization over the disturbance sequence $\mathbf{w}$. Since \eqref{eq:regret_objective_qmi} ensures that $\mathbf{L} \succeq 0$, we can then write the explicit solution to the inner minimization following \cite[Sec. A.5.5]{boyd2004convex}. In particular, it holds that $\min_{\mathbf{w}} ~  (2\mathbf{N}^\top x_0 + \mathbf{H}_w^\top \bm{\lambda})^\top \mathbf{w} + \mathbf{w}^\top \mathbf{L} \mathbf{w}$ equals $-\frac{1}{4}  (2\mathbf{N}^\top x_0 + \mathbf{H}_w^\top \bm{\lambda})^\top \mathbf{L}^\dagger (2\mathbf{N}^\top x_0 + \mathbf{H}_w^\top \bm{\lambda})$ if the range condition
    \begin{equation}
        \label{eq:range_condition_implementation}
        (\mathbf{I} - \mathbf{L} \mathbf{L}^\dagger) (2\mathbf{N}^\top x_0 + \mathbf{H}_w^\top \bm{\lambda}) = 0\,,
    \end{equation}
    where $\mathbf{L}^\dagger$ denotes the pseudo-inverse of $\mathbf{L}$, holds; otherwise, the problem is unbounded. By directly enforcing \eqref{eq:range_condition_implementation} as constraint and introducing the auxiliary epigraph variable $\tau$, we can then express $\mathtt{Reg}_T^\star(\bm{\Phi}_x, \bm{\Phi}_u, \gamma, x_0)$ as
    \begin{align*}
        &~\min_{\tau, \bm{\lambda} \geq 0} ~ \mathbf{h}_w^\top \bm{\lambda} + \tau\\
        &\st~ (\mathbf{I} - \mathbf{L} \mathbf{L}^\dagger) (2\mathbf{N}^\top x_0 + \mathbf{H}_w^\top \bm{\lambda}) = 0\,,\\
        & ~ -x_0^\top \mathbf{M} x_0 + \frac{1}{4}  (2\mathbf{N}^\top x_0 + \mathbf{H}_w^\top \bm{\lambda})^\top \mathbf{L}^\dagger (2\mathbf{N}^\top x_0 + \mathbf{H}_w^\top \bm{\lambda}) \leq \tau\,.
    \end{align*}
    Letting $\bm{\eta} = \frac{1}{2} \bm{\lambda}$ and applying the Schur complement, one can reformulate the worst-case regret as
    \begin{align*}
        &~\min_{\tau, \bm{\eta} \geq 0} ~ 2\mathbf{h}_w^\top \bm{\eta} + \tau\\
        &\st~ \begin{bmatrix}
            \tau + x_0^\top \mathbf{M} x_0 &  (\mathbf{N}^\top x_0 + \mathbf{H}_w^\top \bm{\eta})^\top\\
            (\mathbf{N}^\top x_0 + \mathbf{H}_w^\top \bm{\eta}) & \mathbf{L}
        \end{bmatrix} \succeq 0\,.
    \end{align*}
    Recalling the definition of $\mathbf{M}$, $\mathbf{N}$, and $\mathbf{L}$, we rewrite the above matrix inequality constraint as
    \begin{equation*}
        \begin{bmatrix}
            \tau + x_0^\top \mathbf{C}_T^{[0,0]} x_0 & \bm{\eta}^\top \mathbf{H}_w + x_0^\top \mathbf{C}_T^{[0, 1:T]} \\
            \mathbf{H}_w^\top \bm{\eta} + \mathbf{C}_T^{[1:T, 0]} x_0 & \gamma^2 \mathbf{I} + \mathbf{C}_T^{[1:T, 1:T]}
        \end{bmatrix}
        -
        \mathbf{O}^\top \mathbf{O}
        \succeq 0\,,
    \end{equation*}
    where $\mathbf{O} = \begin{bmatrix} \bar{\mathbf{S}}^\frac{1}{2} \bm{\Phi}_0 x_0 & \bar{\mathbf{S}}^\frac{1}{2} \bm{\Phi}_w \end{bmatrix}$. The desired expression \eqref{eq:sdp_schur_constraints} then follows by applying once more the Schur complement. In particular, note that if \eqref{eq:sdp_schur_constraints} is satisfied, then \eqref{eq:regret_objective_qmi} and \eqref{eq:regret_lmi_maximization_concave} are guaranteed to hold since all principle sub-matrices of a positive semidefinite matrix are positive semidefinite.

    To ensure that the safety constraints are robustly satisfied, we proceed as in \cite{martin2022safe} and employ strong duality of linear optimization problems to eliminate the universal quantifier from \eqref{eq:safety_constraints_compact_Phi}. Specifically, we first note that \eqref{eq:safety_constraints_compact_Phi} holds if and only if $\mathbf{H}_z \bm{\Phi}_w \mathbf{w} \leq \mathbf{h}_z - \mathbf{H}_z \bm{\Phi}_0 x_0$ for all $\mathbf{w} \in \mathbfcal{W}$. Then, we observe that
    \begin{alignat*}{3}
        \max_{\mathbf{w} \in \mathbfcal{W}} ~ \left(
        \mathbf{H}_z
        \bm{\Phi}_w
        \right)_i \mathbf{w} = & ~ \min_{\mathbf{y}_i \geq 0} ~ \mathbf{h}_w^\top \mathbf{y}_i\,,\\
        & \st ~ \mathbf{H}_w^\top \mathbf{y}_i = \left(
        \mathbf{H}_z
        \bm{\Phi}_w
        \right)_i^\top\,,
    \end{alignat*}
    where $\mathbf{y}_i$ denotes the dual vector corresponding with the $i$-th row of the maximization. From this dual reformulation, the set of linear constraints \eqref{eq:sdp_dual_safety_constraints_1}-\eqref{eq:sdp_dual_safety_constraints_2} can be derived by concatenating the dual variables that arise from each row of the maximization into the matrix $\mathbf{Y}$.

\subsection{Proof of Proposition~\ref{prop:implementation_bounded_energy}}
    Let $\tilde{\mathbf{L}} = \mathbf{C}_T^{[1:T, 1:T]} - \bm{\Phi}_w^\top \bar{\mathbf{S}} \bm{\Phi}_w$ so that the inner maximization problem in \eqref{prob:safe_regret_optimal_original_bounded_energy} equivalently reads as
    \begin{equation*}
        - \min_{\norm{\mathbf{w}}_2 \leq \sigma} ~ x_0^\top \mathbf{M} x_0 + 2x_0^\top \mathbf{N}^\top \mathbf{w} + \mathbf{w}^\top \tilde{\mathbf{L}} \mathbf{w}\,.
    \end{equation*}
    Then, since $\mathbf{w} = 0$ verifies the Slater's constraint qualification condition, strong duality holds \cite[App. B]{boyd2004convex}. Hence, by taking the dual we rewrite this optimization problem as
    \begin{align*}
        & ~ \min_{\lambda \geq 0, \gamma} ~ \gamma\\
        & \st ~
        \begin{bmatrix}
            x_0^\top \mathbf{M} x_0 + \gamma - \sigma \lambda & x_0^\top \mathbf{N}\\
            \mathbf{N}^\top x_0 & \lambda \mathbf{I} + \tilde{\mathbf{L}}
        \end{bmatrix} \succeq 0\,.
    \end{align*}
    Recalling the definition of $\mathbf{M}$, $\mathbf{N}$, and $\bar{\mathbf{L}}$, we finally express the above matrix inequality constraint as
    \begin{equation*}
        \begin{bmatrix}
            x_0^\top \mathbf{C}_T^{[0,0]} x_0 + \gamma - \sigma \lambda & x_0^\top \mathbf{C}_T^{[0, 1:T]} \\
            \mathbf{C}_T^{[1:T, 0]} x_0 & \lambda \mathbf{I} + \mathbf{C}_T^{[1:T, 1:T]}
        \end{bmatrix}
        -
        \mathbf{O}^\top \mathbf{O}
        \succeq 0\,,
    \end{equation*}
    from which the constraint \eqref{eq:sdp_schur_constraints_bounded_energy} can be derived by directly applying the Schur complement.

\end{document}